\documentclass[a4paper,USenglish]{lipics-v2021}
\pdfoutput=1
\nolinenumbers
\hideLIPIcs
\usepackage[noend]{algpseudocode}
\usepackage[many]{tcolorbox}
\usepackage{mathtools}
\usepackage{bm}
\usepackage{threeparttable}
\usepackage{tabularx}
\usepackage{url}

\makeatletter \newcommand{\apxthm}[3]{
\newcounter{sec#2} \newcounter{val#2} \newcounter{use#2}
\setcounter{sec#2}{\value{section}} \setcounter{val#2}{\value{theorem}}
\begin{#1} \label{#2} {#3} \end{#1}
\long\@namedef{#2}{
\gdef\thesection{\@arabic\c@section}
\@tempcnta=\value{section} \@tempcntb=\value{theorem}
\setcounter{section}{\value{sec#2}} \setcounter{theorem}{\value{val#2}}
\phantomsection \label{#2*} \begin{#1} {#3} \end{#1} 
\setcounter{section}{\@tempcnta} \setcounter{theorem}{\@tempcntb}
\gdef\thesection{\@Alph\c@section}
}}\makeatother

\newenvironment{algobox}[2][]{
    \begin{center}
    \begin{tcolorbox}[enhanced,title=\centering \large {#2},colback=white,colframe=black!80,width=\textwidth,breakable]
    \underline{\textbf{Code for a party $P_i$}{#1}}
    \begin{algorithmic}[1]
}{
    \end{algorithmic} 
    \end{tcolorbox} 
    \end{center}
}

\algrenewcommand{\algorithmiccomment}[1]{\hfill \textcolor{red}{$\triangleright$ \emph{#1}}}

\makeatletter \newcommand{\algmargin}{\the\ALG@thistlm} \makeatother
\newlength{\whilewidth} 
\algdef{SE}[parWHILE]{parWhile}{EndparWhile}[1]
{\parbox[t]{\dimexpr\linewidth-\algmargin}{
\hangindent\whilewidth\strut\algorithmicwhile\ #1\ \algorithmicdo\strut}}{\algorithmicend\ \algorithmicwhile}
\algnewcommand{\parState}[1]{\State \parbox[t]{\dimexpr\linewidth-\algmargin}{\strut #1\strut}}

\algblock{Upon}{EndUpon} \algrenewtext{Upon}[1]{\textbf{upon} {#1} \textbf{do}}
\makeatletter \ifthenelse{\equal{\ALG@noend}{t}} {\algtext*{EndUpon}} \makeatother
\algblock{When}{EndWhen} \algrenewtext{When}[1]{\textbf{when} {#1} \textbf{do}}
\makeatletter \ifthenelse{\equal{\ALG@noend}{t}} {\algtext*{EndWhen}} \makeatother

\newcommand{\llabel}[1]{\phantomsection\label{line:#1}}

\DeclarePairedDelimiter{\floor}{\lfloor}{\rfloor}
\DeclarePairedDelimiter{\ceil}{\lceil}{\rceil}
\newcommand{\ECHO}{\texttt{ECHO}}
\newcommand{\PROP}{\texttt{PROP}}
\newcommand{\READY}{\texttt{READY}}
\newcommand{\msgpair}[2]{\langle {#1}, {#2} \rangle}
\newcommand{\BO}{\mathcal{O}}
\newcommand{\Term}{\mathsf{Term}}
\newcommand{\TC}{\mathsf{TC}}
\newcommand{\GC}{\mathsf{GC}}
\newcommand{\Exp}{\mathsf{Exp}}
\newcommand{\Agn}{\mathsf{Agr}_\mathsf{N}}
\newcommand{\Agz}{\mathsf{Agr}_\mathsf{Z}}
\newcommand{\TStep}{\mathsf{2\textsf{-}Step}}
\newcommand{\Prop}{\mathsf{Prop}}
\newcommand{\LEFT}{\texttt{LEFT}}
\newcommand{\RIGHT}{\texttt{RIGHT}}
\newcommand{\CENTER}{\texttt{CENTER}}
\newcommand{\KVAL}{\texttt{KVAL}}
\newcommand{\sfnext}{\mathsf{next}}

\title{Asynchronous Approximate Agreement with Quadratic Communication}

\author{Mose {Mizrahi Erbes}}{ETH Zurich, Switzerland}{mmizrahi@ethz.ch}{https://orcid.org/0009-0009-9771-0845}{}
\author{Roger Wattenhofer}{ETH Zurich, Switzerland}{wattenhofer@ethz.ch}{https://orcid.org/0000-0002-6339-3134}{}
\authorrunning{M.\ Mizrahi Erbes and R.\ Wattenhofer}

\Copyright{Mose Mizrahi Erbes, and Roger Wattenhofer}
\begin{document}

\maketitle

\begin{abstract}
We consider an asynchronous network of $n$ message-sending parties, up to $t$ of which are byzantine. We study approximate agreement, where the parties obtain approximately equal outputs in the convex hull of their inputs. In their seminal work, Abraham, Amit and Dolev [OPODIS '04] solve this problem \linebreak in $\mathbb{R}$ with the optimal resilience $t < \frac{n}{3}$ with a protocol where each party reliably broadcasts a value in every iteration. This takes $\Theta(n^2)$ messages per reliable broadcast, or $\Theta(n^3)$ messages per iteration.

In this work, we forgo reliable broadcast to achieve asynchronous approximate agreement against $t < \frac{n}{3}$ faults with a quadratic communication. In a tree with the maximum degree $\Delta$ and the centroid \linebreak decomposition height $h$, we achieve edge agreement in at most $6h + 1$ rounds with $\mathcal{O}(n^2)$ messages of size $\mathcal{O}(\log \Delta + \log h)$ per round. We do this by designing a 6-round multivalued 2-graded consensus protocol and using it to recursively reduce the task to edge agreement in a subtree with a smaller centroid decomposition height. Then, we achieve edge agreement in the infinite path $\mathbb{Z}$, again with the help of 2-graded consensus. Finally, we show that our edge agreement protocol enables $\varepsilon$-agreement in $\mathbb{R}$ in $6\log_2\frac{M}{\varepsilon} + \mathcal{O}(\log \log \frac{M}{\varepsilon})$ rounds with $\mathcal{O}(n^2 \log \frac{M}{\varepsilon})$ messages and $\mathcal{O}(n^2\log \frac{M}{\varepsilon}\log \log \frac{M}{\varepsilon})$ bits of communication, where $M$ is the maximum non-byzantine input magnitude.
\end{abstract}

\section{Introduction} \label{sec1}

We consider a fully connected asynchronous network of $n$ message-passing parties $P_1,\dots,P_n$. Up to $t$ of these parties are corrupted in a byzantine manner, while the rest are honest.

In an approximate (convex) agreement problem, the parties output approximately equal values in the convex hull of their inputs. The most classical example is approximate agreement in $\mathbb{R}$, where the inputs/outputs are in $\mathbb{R}$, and for some parameter $\varepsilon > 0$ the following hold: \begin{itemize}
    \item \textbf{\emph{validity:}} Each honest party output is between the minimum and \mbox{maximum honest inputs.}
    \item \textbf{\emph{$\bm{\varepsilon}$-agreement:}} If any honest parties $P_i$ and $P_j$ output $y_i$ and $y_j$, then $|y_i - y_j| \leq \varepsilon$.
\end{itemize}

Approximate agreement in $\mathbb{R}$ was introduced in 1985 by Dolev, Lynch, Pinter, Stark and Weihl \cite{dolev86}. Like byzantine agreement, in synchronous networks it is possible against $t < \frac{n}{2}$ faults with setup \cite{glw22}, but only possible when $t < \frac{n}{3}$ if perfect (signature-free) security is desired \cite{dolev86} or the network is asynchronous. What sets approximate agreement apart is that it is determinism-friendly. While deterministic byzantine agreement takes $t + 1$ rounds in synchrony \cite{ds83} and is impossible with just one crash in asynchrony \cite{flp85}, approximate agreement does not \linebreak share these limitations. Thus, approximate agreement protocols are customarily deterministic.

In \cite{dolev86}, Dolev et al.\ achieve $\varepsilon$-agreement in $\mathbb{R}$ with a perfectly secure synchronous protocol secure against $t < \frac{n}{3}$ corruptions. Simplifying things slightly, in their protocol the parties estimate the spread $S$ of their inputs (the maximum difference between any two inputs), and run for $\ceil{\log_2\frac{S}{\varepsilon}}$ rounds. In each round each party sends its value to every other party, and with this the parties halve the diameter of their values. After $\ceil{\log_2\frac{S}{\varepsilon}}$ rounds, the spread is at most $2^{-\ceil{\log_2(S / \varepsilon)}} \leq \frac{\varepsilon}{S}$ of what it initially was, and thus $\varepsilon$-agreement is achieved. They present an asynchronous version of this protocol as well, but only with the resilience $t < \frac{n}{5}$ as the parties can no longer wait to receive the value of each honest party in every iteration.

Asynchronous approximate agreement in $\mathbb{R}$ with the optimal resilience $t < \frac{n}{3}$ was first achieved in 2004 by Abraham, Amit and Dolev \cite{aad04}, with a protocol that consists of $\BO(\log\frac{S}{\varepsilon})$ constant-round iterations. Each iteration involves one witness technique application where each party reliably broadcasts its current value in $\mathbb{R}$ (with a reliable broadcast protocol such \linebreak as Bracha's \cite{b87}), and obtains at least $n - t$ reliably broadcast values. The technique ensures that every two parties obtain the values of at least $n - t$ common parties. Since the technique's introduction in \cite{aad04}, most asynchronous approximate agreement protocols have depended on it. Some examples are \cite{aad04,glw22} for agreement in $\mathbb{R}$, \cite{dsz22,glw23,hm13,vg13} for agreement in $\mathbb{R}^d$ when $d \geq 2$, and \cite{cgwf24,nr19} for agreement in graphs (trees, chordal graphs, cycle-free semilattices).

The witness technique requires the parties to reliably broadcast their inputs. Since reliable broadcast requires $\Omega(n^2)$ messages for deterministic \cite{dr85} or strongly adaptive \cite{a19} security against $t = \Omega(n)$ faults, the witness technique requires $\Omega(n^3)$ messages to be sent. Hence, the optimally resilient approximate agreement protocol of Abraham, Amit and Dolev \cite{aad04} costs $\Theta(n^3)$ messages per iteration. This is the case despite asynchronous approximate agreement being possible with $\Theta(n^2)$ messages per iteration, as demonstrated by the protocol of Dolev et al.\ \cite{dolev86} which suboptimally tolerates $t < \frac{n}{5}$ faults. Therefore, we ask the following question: \emph{Is there an asynchronous approximate agreement protocol that optimally tolerates $t < \frac{n}{3}$ faults with only a quadratic (proportional to $n^2$) amount of communication?}

In this work, we answer this question affirmatively by abandoning the witness technique. First, we achieve edge agreement (a discrete form of approximate agreement) in finite trees \cite{nr19} with the optimal resilience $t < \frac{n}{3}$ via multivalued 2-graded consensus iterations. Then, we extend our protocol to achieve edge agreement in the infinite path $\mathbb{Z}$. Finally, we show that edge agreement in $\mathbb{Z}$ implies $\varepsilon$-agreement in $\mathbb{R}$ by reducing the latter to the former. Our final protocol for $\varepsilon$-agreement in $\mathbb{R}$ takes $6\log_2\frac{M}{\varepsilon} + \BO(\log\log \frac{M}{\varepsilon})$ rounds (where $M$ is the maximum honest input magnitude), with $\BO(n^2)$ messages of \mbox{size $\BO(\log\log\frac{M}{\varepsilon})$ sent per round.}

Our work is inspired by \cite{glw24}, which achieves exact convex agreement in $\mathbb{Z}$ with byzantine agreement iterations in a synchronous network. We instead achieve edge agreement in $\mathbb{Z}$ with graded consensus, which is much simpler than byzantine agreement, especially in asynchronous \linebreak networks. Note though that \cite{glw24} supports large inputs with less communication than us.

\section{Model \& Definitions} \label{sec_model}

We consider an asynchronous network of $n$ message-sending parties $P_1,P_2,\dots,P_n$, fully connected via reliable and authenticated channels. An adversary corrupts up to $t < \frac{n}{3}$ parties, making them byzantine, and these parties become controlled by the adversary. The adversary adaptively chooses the parties it wants to corrupt during protocol execution, depending on the messages sent over the network. If a party is never corrupted, then we call it honest.

The parties do not have synchronized clocks. The adversary can schedule messages as it sees fit, and it is only required to eventually deliver messages with honest senders. If a party sends a message, then the adversary may corrupt the party instead of delivering the message.

We say that a party multicasts $m$ when it sends $m$ to every party. By corrupting a party that is multicasting a message, the adversary may deliver the message to only some parties.

Our protocols are live; i.e.\ they achieve liveness. That is, if the honest parties all acquire inputs and all run forever, then they all output. However, in Section \ref{sec_termination}, we explain a low-cost way to upgrade our live protocols into terminating versions \mbox{that allow the parties to halt.}

To define asynchronous round complexity, we imagine an external clock. If a protocol runs in $R$ rounds, then it is live or terminating, and the time elapsed between when every honest party running the protocol knows its input and when every honest party outputs/terminates is at most $R\Delta$, where $\Delta$ is the maximum honest message delay in the protocol's execution.

To keep mathematical expressions simple, we use the convention that $\log v = 0$ if $v \leq 0$.

\subsubsection*{Edge Agreement in a Tree} In edge agreement in a tree graph $T = (V, E)$, each party $P_i$ acquires an input vertex $v_i \in V$, and outputs a vertex $y_i \in V$. We want the following properties: \begin{itemize}
    \item \textbf{\emph{edge agreement:}} Every two honest output vertices are either equal or adjacent in $T$.
    \item \textbf{\emph{convex validity:}} For every honest output $y$, there exist some (possibly equal) honest inputs $v_y$ and $v'_y$ such that $y$ is on the path which connects $v_y$ and $v'_y$ in $T$.
\end{itemize}

Edge agreement in a tree generalizes edge agreement in a path, which is essentially the same task as approximate agreement in an interval in $\mathbb{R}$, but with the input/output domain restricted to the integers (with adjacent integers representing adjacent path vertices).

\subsubsection*{Graded Consensus} In $k$-graded consensus, each party $P_i$ acquires an input $v_i$ in an input domain $\mathcal{M}$, and outputs some value-grade pair $(y_i, g_i) \in (\mathcal{M} \times \{1, \dots, k\}) \cup \{(\bot, 0)\}$. The following must hold: \begin{itemize}
    \item \textbf{\emph{agreement:}} If any honest parties $P_i$ and $P_j$ output $(y_i, g_i)$ and $(y_j, g_j)$, then $|g_i - g_j| \leq 1$, and if $\min(g_i, g_j) \geq 1$, then $y_i = y_j$.
    \item \textbf{\emph{intrusion tolerance:}} If $(y, g) \neq (\bot, 0)$ is an honest output, then $y$ is an honest input.
    \item \textbf{\emph{validity:}} If the honest parties have a common input $m \in \mathcal{M}$, then they all output $(m, k)$.
\end{itemize}

As \cite{aw24} has noted before, $k$-graded consensus is equivalent to edge agreement in a particular tree when the inputs must all be leaves of the tree.

\begin{figure}[ht]
    \centering
    \begin{tikzpicture}[scale=1]
        \node[] at (1, 0) (a3) {$(a, 2)$};
        \node[] at (2.5, 0) (a2) {$(a, 1)$};
        \node[] at (4, 0) (bot2) {$(\bot, 0)$};
        \node[] at (4, 1) (b2) {$(b, 1)$};
        \node[] at (4, 2) (b3) {$(b, 2)$};
        \node[] at (5.5, 0) (c2) {$(c, 1)$};
        \node[] at (7, 0) (c3) {$(c, 2)$};
        \draw[] (bot2)--(a2)--(a3);
        \draw[] (bot2)--(b2)--(b3);
        \draw[] (bot2)--(c2)--(c3);
    \end{tikzpicture}
    \caption{Edge agreement in a spider tree with the center $(\bot, 0)$ and a path $((m, 1), \dots, (m, k))$ attached to it for each $m \in \mathcal{M}$ is equivalent to $k$-graded consensus when the parties can only have leaf edge agreement inputs, with each leaf input $(m, k)$ \mbox{in bijection with the $k$-graded consensus input $m$.}}\label{gradedfig}
\end{figure}

In this work, we use both binary 2-graded consensus (with $|\mathcal{M}| = 2$) and multivalued 2-graded consensus (with $|\mathcal{M}| > 2$). For the latter, there is a 9-round protocol in the literature that costs $\BO(n^2)$ messages of size $\BO(\log |\mathcal{M}|)$ \cite{aw24}. This protocol is enough for our asymptotic complexity results. However, it achieves a property called binding \cite{aw24} that we do not need, and without this property, 6 rounds suffice. We show this in the \hyperref[GradedConsensus]{appendix} by constructing a family of multivalued $2^k$-graded consensus protocols $\mathsf{GC}_{2^0}, \mathsf{GC}_{2^1}, \mathsf{GC}_{2^2}, \mathsf{GC}_{2^3}, \dots$ which each take $3k + 3$ rounds, with $\BO(n^2)$ messages of size $\BO(\log k + \log |\mathcal{M}|)$ per round. We obtain this family by constructing a $1$-graded consensus protocol $\mathsf{GC}_1$, \mbox{and by repeatedly grade-doubling it.}\footnote{Repeated grade-doubling is a standard method to achieve $2^k$-graded consensus in $\BO(k)$ rounds \cite{fmll21, delphi, simon25}, though note that synchronous networks allow $k^k$-graded consensus in $\BO(k)$ rounds \cite{ggl22}.}

\section{Overview \& Contributions} \label{sec_overview}

Our first contribution is a protocol for edge agreement in finite trees. Against byzantine faults, this problem was first studied by Nowak and Rybicki \cite{nr19}. It generalizes both edge agreement in finite paths (the discrete version of $\varepsilon$-agreement in $[0, 1]$) and graded consensus. 

Nowak and Rybicki achieve edge agreement in a finite tree $T = (V, E)$ of diameter $D$ with $\ceil{\log_2 D} + 1$ constant-round witness technique iterations and thus $\Theta(n^3\log D(\log |V| + \log n))$ bits of communication, where the $\log n$ term arises from the party IDs that identify each witness technique reliable broadcast's sender party. Meanwhile, we achieve edge agreement with at most $h(T)$ iterations ($6h(T) + 1$ rounds), where $h(T)$ is a tree property we formally define in Section \ref{sec_tree} that upper bounds the height of $T$'s centroid decompositions \cite{centroid-decomp}. The integer value $h(T)$ can be anywhere in $[\floor{\log_2 D}, \ceil{\log_2|V|}]$, which means that our protocol's round complexity is for some trees (though not for spider trees, trees with $\BO(D)$ vertices such as paths etc.) worse than Nowak and Rybicki's. However, our iterations only cost $\BO(n^2)$ messages, each of size most $\BO(\log \Delta + \log(h(T)))$ where $\Delta$ is $T$'s maximum degree. So, our protocol requires roughly $n$ times less communication when $h(T) \approx \log_2 D$.

In Section \ref{sec_tree}, we present a parametrized recursive protocol $\TC(T)$ that achieves edge agreement in any given finite tree $T$. On a high level, it works as follows:
\begin{enumerate}
    \item If $T$ has 1 or 2 vertices, then each party outputs its input vertex. This is the base case.
    \item If $T$ has $s \geq 3$ vertices, then the parties let $\sigma$ be a centroid vertex of $T$ (whose deletion from $T$ results in a forest whose components all have at most $s/2$ vertices), and let $w_1, \dots, w_d$ be $\sigma$'s neighbors sorted by vertex index. Then, they run 2-graded consensus, where each party's input is either $\sigma$ (if its edge agreement input is $\sigma$) or some neighbor $w_k$ of $\sigma$ (if its edge agreement input is in $H_k$, which is how we refer to the tree component of $T \setminus \{\sigma\}$ that contains $w_k$). In the parties reach consensus on $\sigma$, then they output $\sigma$. Otherwise, if they reach consensus on some neighbor $w_k$ of $\sigma$, then the parties with input vertices outside $H_k$ adopt the new input $w_k$, and we reduce \mbox{the task to edge agreement in the subtree $H_k$.}
\end{enumerate}

There is a snag. The explanation above only works if the parties actually reach unanimous agreement on either $\sigma$ or one of its neighbors $w_k$. However, 2-graded consensus does not guarantee this, since some parties might output $(\bot, 0)$ from it. What allows us to overcome this issue is that if anybody outputs $(\bot, 0)$, then the parties all learn that they ran 2-graded consensus with differing inputs, and thus learn that $\sigma$ is a \mbox{safe output vertex w.r.t\ convex validity.}

Our approach for finite trees above corresponds to binary search when the tree is a path. For example, the parties reach edge agreement in the path $(0, \dots, 8)$ by either directly agreeing on $4$, or by reducing the problem to edge agreement in either $(0, \dots 3)$ or $(5, \dots, 8)$.  Binary search does not support the infinite path $\mathbb{Z}$. Fortunately, 2-graded consensus also enables exponential search. In Section \ref{sec_z}, we present a protocol for edge agreement in $\mathbb{Z}$ that on a high level works as follows (with some complications that we skip over for now due to our use of 2-graded consensus instead of byzantine agreement).\begin{enumerate}
    \item First, the parties reach 2-graded consensus on whether they prefer to agree on the left path $(\dots, -1, 0)$ or the right path $\mathbb{N} = (0, 1, \dots)$, with each party preferring $\mathbb{N}$ iff its input is in $\mathbb{N}$. Below, we explain what the parties then do if they decide to agree in $\mathbb{N}$. Otherwise, they follow the same steps, but with mirrored \mbox{(sign-flipped) inputs and outputs.}
    \item The parties run exponential search with the phases $k = 0, 1, \dots$; where in each phase $k$ they reach 2-graded consensus on if they have inputs in the left path $(2^k - 1, \dots, 2^{k+1} - 1)$ or the right path $(2^{k+1}, \dots)$. If they decide on the left path, then they reach edge agreement \linebreak in it using our protocol for edge agreement in finite trees. Otherwise, if they decide on the right path, then they increment the \mbox{phase counter $k$ and continue exponential search.}
    \item Instead of directly using exponential search for edge agreement in $\mathbb{N}$, we take inspiration from \cite{bentley76}, and design a two-stage protocol that is asymptotically twice as round-efficient. Roughly speaking, the parties run the protocol we described above based on exponential search to approximately agree on some $k$ such that the path $(2^k - 1, \dots, 2^{k+1} - 1)$ contains safe output values w.r.t\ convex validity, and then they reach edge agreement in this path.
\end{enumerate}

When the maximum honest input magnitude is $M$, our protocol for edge agreement in $\mathbb{Z}$ takes $6\log_2 M + \BO(\log \log M)$ rounds, with $\BO(n^2)$ messages of size $\BO(\log \log M)$ per round. In Section \ref{sec_reduction}, we reduce $\varepsilon$-agreement in $\mathbb{R}$ to edge agreement in $\mathbb{Z}$ to show that this implies $\varepsilon$-agreement in $\mathbb{R}$ in $6\log_2\frac{M}{\varepsilon} + \BO(\log \log \frac{M}{\varepsilon})$ rounds with $\BO(n^2\log \frac{M}{\varepsilon})$ messages and $\BO(n^2\log \frac{M}{\varepsilon}\log \log \frac{M}{\varepsilon})$ bits of communication in total. The reason why we use this reduction instead of directly solving $\varepsilon$-agreement is that using the reduction makes termination easier. Note that factor 6 in the round complexity is due to us using our 6-round 2-graded consensus protocol. If for any other $r$ we used a $r$-round protocol instead, like the 2-round 2-graded consensus protocol in \cite{aw24} that tolerates $t < \frac{n}{5}$ faults, then this factor would be $r$ instead.

In terms of message and communication (though not round) complexity, our protocol for $\varepsilon$-agreement in $\mathbb{R}$ is more efficient than that of Abraham et al.\ \cite{aad04}, who achieve $\varepsilon$-agreement in $\mathbb{R}$ with $\BO(\log\frac{S}{\varepsilon})$ constant-round witness technique iterations (where $S$ is the honest input spread, i.e.\ the maximum difference of any honest inputs), and \mbox{with $\Theta(n^3\log \frac{S}{\varepsilon})$ messages in total.}

Another notable protocol is Delphi, by Bandarupalli, Bhat, Bagchi, Kate, Liu-Zhang and Reiter \cite{delphi}. To efficiently achieve $\varepsilon$-agreement with $\ell$-bit inputs in $\mathbb{R}$, they assume an input distribution (normal distribution for the following), and when the honest input spread is $S$ they achieve $\varepsilon$-agreement except with probability $2^{-\lambda}$ in $\BO(\log(\frac{S}{\varepsilon}\log\frac{S}{\varepsilon}) + \log(\lambda\log n))$ rounds with $\BO(\ell n^2 \frac{S}{\varepsilon}(\log(\frac{S}{\varepsilon}\log\frac{D}{\varepsilon}) + \log(\lambda\log n)))$ bits of communication, while relaxing validity by allowing outputs outside the range of the honest inputs by at most $S$. They use the security parameter $\lambda$ here to assume bounds on $S$ that hold except with $2^{-\lambda}$ probability thanks to their input distribution assumptions. In comparison, we achieve $\varepsilon$-agreement in $\mathbb{R}$ without relaxing validity or assuming any input bounds. As Table \ref{comparisons} shows, our protocol is also more efficient, in particular since Delphi requires a cubic amount of \mbox{communication per round when $S \geq n \cdot \varepsilon$.}

\begin{table}[ht] \begin{threeparttable}
    \renewcommand{\tnote}[1]{\textsuperscript{#1}}
    \newcolumntype{C}{>{\centering\arraybackslash}c}
    \newcolumntype{Y}{>{\centering\arraybackslash}X}
    \caption{Comparison of protocols for asynchronous $\varepsilon$-agreement in $\mathbb{R}$ when the parties have inputs in $[0, 1]$. If $v_\mathsf{lo}$ and $v_\mathsf{hi}$ are the minimum and maximum honest inputs, then $S = v_\mathsf{hi} - v_\mathsf{lo}$ and $M = v_\mathsf{hi}$. To make the comparisons simple, we assume for \cite{dolev86}, \cite{aad04} and \cite{delphi} that the inputs are all multiples of $\varepsilon$.}
    \begin{tabularx}{1\columnwidth}{|C|C|C|C|Y|} \hline
        \textbf{Threshold} & \textbf{Bits Sent / Round} & \textbf{Round Complexity} & \textbf{Relaxation}\tnote{a} & \textbf{Source}
        \\\hline
        $t < \frac{n}{5}$ & $\BO(n^2\log\frac{1}{\varepsilon})$ & $\ceil{\log_2\frac{1}{\varepsilon}}$ & 0 & \cite{dolev86}
        \\\hline
        $t < \frac{n}{3}$ & $\BO(n^3\log\frac{n}{\varepsilon})$\tnote{b} & $\BO(\log\frac{S}{\varepsilon})$ & 0 & \cite{aad04}
        \\\hline
        $t < \frac{n}{3}$ & $\BO(n^2\min(\frac{S}{\varepsilon}, n\log\frac{1}{\varepsilon}))$ & $\BO(\log(\frac{\log(1/\varepsilon)\min(1/\varepsilon, n)}{\varepsilon}))$ & $S$ & \cite{delphi}
        \\\hline
        \textcolor{blue}{$t < \frac{n}{3}$} & \textcolor{blue}{$\BO(n^2\log\log\frac{M}{\varepsilon})$\tnote{c}} & \textcolor{blue}{$\BO(\log \frac{M}{\varepsilon})$} & \textcolor{blue}{0} & \textcolor{blue}{this work}
        \\\hline
    \end{tabularx}
    \begin{tablenotes}
        \item[a] The relaxation is how far an honest output is allowed to be from the honest input range $[v_\mathsf{lo}, v_\mathsf{hi}]$.
        \item[b] The first few rounds of \cite{aad04} estimate the spread $S$, and this costs $\Theta(n^4\log\frac{1}{\varepsilon})$ bits of communication. However, this can be reduced to $\Theta(n^3\log\frac{n}{\varepsilon})$ with modern reliable broadcast protocols \cite{chen24}.
        \item[c] The $\log\log\frac{M}{\varepsilon}$ factor here is for tags that distinguish messages sent in different protocol iterations.
    \end{tablenotes}
    \label{comparisons}
\end{threeparttable} \end{table}

\section{Edge Agreement in a Tree} \label{sec_tree}
In $\mathbb{R}^d$, a set $Z \subseteq \mathbb{R}^d$ is straight-line convex if for all $z_1, z_2 \in Z$ it contains the line segment (the shortest path in $\mathbb{R}^d$) that connects $z_1$ and $z_2$. This definition translates to convexity on a tree $T = (V, E)$, where a set $Z \subseteq V$ is convex if for all $z_1, z_2 \in Z$ the set $Z$ contains all the vertices on the shortest path (the only path) between $z_1$ and $z_2$ \cite{nr19, cgwf24}. Hence, we can define convex hulls on $T$, where the convex hull $\langle Z \rangle$ of any vertex set $Z \subseteq V$ is the set that consists of the vertices in $Z$ and every other vertex $v \in V$ that \mbox{is on the path between some $z_1, z_2 \in V$.}

For edge agreement, the convex validity property is that when the honest parties have the set of inputs $X$, they obtain outputs in $\langle X \rangle$. To achieve this, we rely on the following fact: \begin{proposition} \label{convexprop}
    For every tree $T = (V, E)$, $Z \subseteq V$ and $Y \subseteq \langle Z \rangle$ \mbox{it holds that $\langle Y \rangle \subseteq \langle Z \rangle$.}
\end{proposition}

\begin{proof}
Observe that for any tree $T = (V, E)$ and any $Z \subseteq V$, the graph $T[\langle Z \rangle]$ (the subgraph of $T$ induced by $\langle Z \rangle$) is a connected union of paths from $T$, which makes $T[\langle Z \rangle]$ a tree itself. For any $Y \subseteq \langle Z \rangle$, the tree $T[\langle Z \rangle]$ contains every $y_1, y_2 \in Y$, and as $T[\langle Z \rangle]$ is a tree it also contains the path which connects $y_1$ and $y_2$. That is, $T[\langle Z \rangle]$ contains every vertex in $\langle Y \rangle$.
\end{proof}

Every finite tree $T$ has a set of \emph{centroid} vertices such that if one deletes a centroid vertex $\sigma$ from $T$, then every component tree of the resulting forest $T \setminus \{\sigma\}$ has at most half as many vertices as $T$ \cite{centroid-decomp}. Note that a centroid cannot be a leaf if $T$ has 3 or more vertices, since deleting a leaf from a tree leaves behind a connected tree with only one less vertex.

For any finite tree $T$, one can recursively define a centroid decomposition of $T$ to be a rooted tree $T'$ with the following properties (visualized in \cite{centroid-decomp}):\begin{itemize}
    \item The root of $T'$ is a centroid vertex $\sigma$ of $T$.
    \item If in $T$ the centroid $\sigma$ has exactly $d$ neighbors $w_1, \dots, w_d$, then in $T'$ the root $\sigma$ has exactly $d$ child subtrees $W_1, \dots, W_d$, such that for all $k \in \{1, \dots, d\}$ the child subtree $W_k$ is a centroid decomposition of the tree component that contains $w_k$ in the forest $T \setminus \{\sigma\}$.
\end{itemize}

Finally, let us define the \emph{centroid decomposition height} $h(T)$ of a finite tree $T = (V, E)$ to be the maximum height of any centroid decomposition of $T$. The recursive definition of a centroid decomposition above allows one to prove by induction that $h(T) \leq \ceil{\log_2 |V|}$ (where $|V|$ is the number of vertices), and this bound is tight if $T$ is a path and $|V|$ is a power of $2$. However, some trees have low centroid decomposition heights despite having many vertices. For example, if $T$ is a star, then $h(T) = 1$, no matter how many vertices $T$ has.

Below, we present a recursive protocol $\TC(T)$ based on centroid decomposition for edge agreement in a finite tree $T$. If $T$ has at most two vertices, then each party $P_i$ just outputs its input vertex. Otherwise, the parties let $\sigma$ be the minimum-index centroid of $T$, and run our 2-graded consensus protocol $\GC_2$ to either directly output $\sigma$, or to reduce edge agreement in $T$ to edge agreement in a component tree $T'$ of $T \setminus \{\sigma\}$, handled with a recursive $\TC(T')$ instance. The recursion depth is at most $h(T)$ since each recursive call represents a step from a vertex to its child in a centroid decomposition \mbox{of $T$ whose height is upper bounded by $h(T)$.}

\begin{algobox}[\textbf{ with the input} $(v_i)$]{Protocol $\TC(T)$}
    \renewcommand{\thempfootnote}{\fnsymbol{mpfootnote}}
    \setcounter{mpfootnote}{1}
    \If{$T$ has 1 or 2 vertices}
        \State output $v_i$ and do not run the rest of the protocol
    \EndIf
    \State let $\sigma$ be the minimum-index centroid vertex of $T$, of degree $d \geq 2$
    \State let $w_1, \dots, w_d$ be the neighbors of $\sigma$, sorted by vertex index
    \State let $H_j$ be the tree component in $T \setminus \{\sigma\}$ that contains $w_j$ for each $j \in \{1, \dots, d\}$
    \State run an instance of $\GC_2$ with the other parties where your input is $0$ if $v_i = \sigma$, and otherwise your input is the unique index $j$ such that $v_i \in H_j$
    \State wait until you output some $(k, g)$ from $\GC_2$\llabel{tc-gc-wait}
    \If{$k = 0$}
        \State output $\sigma$ \llabel{tc-direct-sigma}
    \ElsIf{$g \geq 1$}
        \State let $v_i^{\sfnext} \gets v_i$ if $g_i = 2 \land v_i \in H_k$, and let $v_i \gets w_k$ otherwise \llabel{tc-set-v}
        \State let $T_i^{\sfnext} \gets H_k$ \llabel{tc-set-t}
        \State if $g = 1$, then multicast $\msgpair{\KVAL}{k}$
    \Else
        \State output $\sigma$
        \State multicast $\CENTER$
        \parState{wait until you have received the message $\msgpair{\KVAL}{k}$ for some $k$ from $t + 1$ parties, and let $(v_i^{\sfnext}, T_i^{\sfnext}) \gets (w_k, H_k)$ when this happens}\llabel{tc-kval-wait}
    \EndIf
    \State run an instance of $\TC(T_i^{\sfnext})$ with the other parties where your input \nolinebreak is $v_i^{\sfnext}$
    \State if $g_i \geq 1$, then when you output $y$ from the $\TC(T_i^{\sfnext})$, output $y$ from $\TC(T)$
    \Upon{receiving $\CENTER$ from $t + 1$ parties}
        \State output $\sigma$ if you haven't output before and ignore future output commands\llabel{tc-out-center}\hspace{2em}       
        \Comment{Note that $P_i$ should run Line \ref{line:tc-out-center} as soon as it receives $\CENTER$ from $t + 1$ parties, even if it is waiting for something else, e.g.\ waiting to output from $\GC_2$ on Line \ref{line:tc-gc-wait}}
    \EndUpon
\end{algobox}

The idea behind $\TC(T)$ when $T$ has $3$ or more vertices is that either there exists some component $H_k$ of $T \setminus \{\sigma\}$ that contains every honest input vertex $v_i$, or there is no such \nolinebreak vertex. \begin{itemize}
    \item In the former case where there is such a vertex, every honest party $P_i$ runs $\GC_2$ with the input $k$, outputs $(k, 2)$ from it, lets $(v_i^\sfnext, T_i^\sfnext) = (v_i, H_k)$, and obtains its final output from a recursive $\TC(H_k)$ instance which it runs with the input $v_i^\sfnext = v_i$. Thus, edge agreement in $T$ is reduced to edge agreement in $H_k$, which the parties reach via $\TC(H_k)$.
    \item In the latter case where there is no such vertex, there are some honest inputs $v_i$ and $v_j$ such that either $\sigma \subseteq \{v_i, v_j\}$ or $v_i$ and $v_j$ are in different components of $T \setminus \{\sigma\}$. So, $\sigma$ is on the path which connects $v_i$ and $v_j$ in $T$, which makes it a safe output vertex w.r.t.\ convex validity. Moreover, if some honest party $P_i$ outputs $(k, g)$ from $\GC_2$ for some $k \not \in \{\bot, 0\}$ and either $g = 1$ or $P_i$'s $\TC(T)$ input $v_i$ is not in $H_k$, then some but not all of the honest parties have inputs in $H_k$, which means that $w_k$ is a safe output vertex w.r.t.\ convex validity since it is incident to the edge that connects $H_k$ with the rest of the tree. With these in mind, we assign each $\GC_2$ output a behavior such that no matter which two adjacent $\GC_2$ outputs the parties settle on, \mbox{they behave in a compatible manner that leads to edge agreement.}
\end{itemize}

\begin{figure}[ht]
    \centering
    \begin{tikzpicture}[scale=0.91]
        \draw [rounded corners=6mm,fill=gray!20] (-4.75,-{sqrt(3)/4})--(-3,{sqrt(3)*1.5})--(-1.25,-{sqrt(3)/4})--cycle;
        \draw [rounded corners=6mm,fill=gray!20] (4.75,-{sqrt(3)/4})--(3,{sqrt(3)*1.5})--(1.25,-{sqrt(3)/4})--cycle;
        \draw [rounded corners=6mm,fill=gray!20] (-1.75,{1.5+sqrt(3)*1.25})--(0,{1.5-sqrt(3)/2})--(1.75,{1.5+sqrt(3)*1.25})--cycle;
        \node[label=below:\Large $H_1$] at (-3, {sqrt(3)*0.625}) (h1) {};
        \node[label=below:\Large $H_3$] at (3, {sqrt(3)*0.625}) (h3) {};
        \node[label=above:\Large $H_2$] at (0, {1.5+sqrt(3)*0.375}) (h2) {};
        \node[draw, circle, minimum size=6mm, label=center:$\sigma$] at (0, 0) (sigma) {};
        \node[draw, circle, minimum size=6mm, label=center:$w_1$] at (-2, 0) (w1) {};
        \node[draw, circle, fill] at (-4, 0) (w11) {};
        \node[draw, circle, fill] at (-3, {sqrt(3)}) (w12) {};
        \node[draw, circle, minimum size=6mm, label=center:$w_3$] at (2, 0) (w3) {};
        \node[draw, circle, fill] at (4, 0) (w31) {};
        \node[draw, circle, fill] at (3, {sqrt(3)}) (w32) {};
        \node[draw, circle, minimum size=6mm, label=center:$w_2$] at (0, 1.5) (w2) {};
        \node[draw, circle, fill] at (-1, {1.5+sqrt(3)}) (w21) {};
        \node[draw, circle, fill] at (1, {1.5+sqrt(3)}) (w22) {};
        \draw[] (w11)--(w1)--(sigma)--(w3)--(w31);
        \draw[] (w12)--(w1);
        \draw[] (w32)--(w3);
        \draw[] (w21)--(w2)--(sigma);
        \draw[] (w22)--(w2);
    \end{tikzpicture}
    \caption{Let $T$ be the tree depicted above with the unique centroid $\sigma$ and the subtrees $H_1$, $H_2$, $H_3$ (the components of $T \setminus \{\sigma\}$). If the parties run $\TC(T)$, then each party with the input $\sigma$ runs $\GC_2$ with the input $0$, and each party $P_i$ with an input $v_i \neq \sigma$ lets its $\GC_2$ input be the index $k \in \{1, 2, 3\}$ such that $v_i \in H_k$. If for some $k \in \{1, 2, 3\}$ every input $v_i$ is in $H_k$, then every party runs $\mathsf{GC}_2$ with the input $k$ and thus outputs $(k, 2)$ from it. Otherwise, $\sigma$ \mbox{is in the convex hull of the input vertices $v_i$.}}
\end{figure}

\begin{theorem} \label{tc-security}
    For any finite tree $T$, suppose the honest parties run $\TC(T)$ with input vertices in $T$. Then, they reach edge agreement in $T$ based on their input vertices in at most $6h(T)$ rounds if they have a common input vertex, and in at most $6h(T) + 1$ rounds otherwise.
\end{theorem}

\begin{proof}
    Below, we show for any finite tree $T$ that if $\TC(H)$ works well (in accordance with the theorem) for every tree $H$ such that $h(H) < h(T)$, then $\TC(T)$ also works well. So, by strong induction on $h(T)$, $\TC(T)$ works well for \mbox{every finite tree $T$, no matter what $h(T) \geq 0$ is.}

    In the base case where $T$ has at most $2$ vertices, the parties reach edge agreement in $0$ rounds by outputting their inputs. This is what happens when $h(T) = 0$, as $h(T) = 0$ iff $T$ is a vertex. For the rest of the proof, suppose $T$ has at least $3$ vertices, which implies $h(T) \geq 1$.

    Let $\sigma$ be the minimum-index centroid vertex of $T$, with the neighbors $w_1, \dots, w_d$ sorted by vertex index and the corresponding tree components $H_1, \dots, H_d$ of $T \setminus \{\sigma\}$ such that $H_j$ contains $w_j$ for all $j$. Observe that $h(T)$ is by definition greater than $h(H_j)$ for all $j$ since $T$ has some centroid decomposition $T'$ with the root vertex $\sigma$ such that $\sigma$ has $d$ child subtrees $W_1, \dots, W_d$ in $T'$, where each subtree $W_j$ is a centroid decomposition of $H_j$ of height $h(H_j)$. So, our inductive assumption tells us \mbox{that $\TC(H_j)$ works well for all $j \in \{1, \dots, d\}$.}

    First, let us consider the simpler scenario, which is when there exists some $k \in \{1, \dots d\}$ such that the subtree $H_k$ contains every honest input vertex. In this scenario, the honest parties all run $\GC_2$ with the input $k$, and thus all output $(k, 2)$ from it. Then, each honest party $P_i$ sets $(v_i^\sfnext, T_i^\sfnext) = (v_i, H_k)$, runs an instance of $\TC(T_i^\sfnext) = \TC(H_k)$ with the other parties where its input is $v_i^\sfnext = v_i$, and obtains its final output from this recursive $\TC(H_k)$ instance. Therefore, edge agreement in $T$ follows from edge agreement in $H_k$, which the parties reach via $\TC(H_k)$. The round complexity of $\TC(T)$ here is that of $\GC_2$ and $\TC(H_k)$ added together; which is always at most $6 + 6h(H_k) + 1 \leq 6h(T) + 1$, and is at most $6 + 6h(H_k) \leq 6h(T)$ if the honest parties run $\TC(T)$ (and thus $\TC(H_k)$) with a common input. Note that the honest parties do not output $\sigma$ on Line \ref{line:tc-out-center} in this scenario since none of them obtain the grade $0$ from $\GC_2$, which means that no honest party multicasts $\CENTER$ and thus that no honest party receives the message $\CENTER$ from $t + 1$ parties.

    Now, let us consider the more complicated scenario where none of the subtrees $H_1,\dots,H_d$ contains every honest input vertex. Then, there exist some distinct honest inputs $v_i$ and $v_j$ such that either $\sigma \in \{v_i, v_j\}$, or $v_i$ and $v_j$ are in different components of $T \setminus \{\sigma\}$. In either case, $\sigma$ is on the path that connects $v_i$ and $v_j$ in $T$, which means that $\sigma$ is in the convex hull of the honest inputs. In addition, if for some $k \in \{1, \dots, d\}$ an honest party $P_i$ outputs $(k, g)$ from $\GC_2$ and either $g = 1$ or $P_i$'s input $v_i$ is not in $H_k$, then some but not all of the honest parties have inputs in $H_k$ (the ``some'' part by $\GC_2$'s intrusion tolerance and the ``not all'' part by either $P_i$'s grade being below $2$ or by $P_i$'s input not being in $H_k$), and this places $w_k$ in the convex hull of the honest inputs since $w_k$ is on every path in $T$ that connects the honest inputs in $H_k$ with the honest inputs outside $H_k$. With these in mind, let us consider all the ways a $\TC(T)$ execution can go, depending on the honest parties' $\GC_2$ outputs.
    \begin{itemize}
        \item It could happen that the honest parties all output $(k, 2)$ or $(k, 1)$ from $\GC_2$ for some $k \in \{1, \dots, d\}$, with at least one outputting $(k, 2)$. This situation is similar to the one we considered previously, except for the fact that some honest parties $P_i$ (those with inputs outside $H_k$ and those with the $\GC_2$ grade $1$) set $v_i^{\sfnext} = w_k$ instead of setting $v_i^{\sfnext} = v_i$. By Proposition \ref{convexprop}, them doing this does not impact convex validity, as the convex hull $\langle V \rangle$ of the honest $\TC(T)$ input vertices $V = \bigcup_{P_i \text{ is honest}}\{v_i\}$ is a superset of the convex hull of $\{w_k\} \cup V \subseteq \langle V \rangle$. So, edge agreement in $T$ follows from edge agreement in $H_k$ (in at most $6 + 6h(H_k) + 1 \leq 6h(T) + 1$ rounds) because after running $\GC_2$, the honest parties all run $\TC(H_k)$ with safe inputs in $H_k$ and all obtain their final outputs from $\TC(H_k)$. Again in this case, \mbox{the parties do not output $\sigma$ on Line \ref{line:tc-out-center} since none of them multicast $\CENTER$.}
        \item It could happen that the honest parties all obtain $\GC_2$ outputs in $\{(\bot, 0),(0, 1),(0, 2)\}$. Then, every honest party outputs $\sigma$, either directly on Line \ref{line:tc-direct-sigma} or Line \ref{line:tc-out-center} after it outputs from $\GC_2$, or even earlier by receiving $t + 1$ $\CENTER$ messages. In either case, the honest parties reach exact agreement on the safe vertex $\sigma$, in at most $6 \leq 6h(T) + 1$ rounds.
        \item Finally, it could happen that the honest parties all output $(k, 1)$ or $(\bot, 0)$ from $\GC_2$ for some $k \in \{1, \dots, d\}$, with at least one outputting $(k, 1)$. Then, every honest party $P_i$ that runs $\TC(T_i^\sfnext)$ does so with the tree $T_i^\sfnext$ set to $H_k$ and with the input vertex $v_i^\sfnext = w_k$. For the honest parties that output $(k, 1)$ from $\GC_2$, this follows from Line \ref{line:tc-set-v} and Line \ref{line:tc-set-t}. Meanwhile, for an honest party that outputs $(\bot, 0)$ and thus sets $(v_i^\sfnext, T_i^\sfnext) = (w_{k'}, H_{k'})$ on Line \ref{line:tc-kval-wait} upon receiving the message $(\KVAL, k')$ from $t + 1$ parties; this follows from the fact that every honest $\KVAL$ message is on $k$, which means that $k$ is the only value on which a party can receive $t + 1$ $\KVAL$ messages. From all of these, we conclude that every honest party either outputs the safe vertex $\sigma$ (on Line \ref{line:tc-direct-sigma} or Line \ref{line:tc-out-center}), or outputs $\sigma$'s safe neighbor $w_k$ after outputting $w_k$ from $\TC(H_k)$ which the honest parties only run with the input $w_k$. Finally, what remains to show is liveness. Observe that since $n - t \geq 2t + 1$, either $t + 1$ honest parties output $(\bot, 0)$ from $\GC_2$, or \mbox{$t + 1$ honest parties output $(k, 1)$ from $\GC_2$.} \begin{itemize}
            \item If the former happens, then $t + 1$ honest parties multicast $\CENTER$ after outputting $(\bot, 0)$ from $\GC_2$. Hence, after one round following $\GC_2$ (i.e.\ $7 \leq 6h(T) + 1$ rounds after $\TC(T)$ begins), every honest party becomes able to output $\sigma$ on Line \ref{line:tc-out-center}.
            \item If the latter happens, then $t + 1$ honest parties multicast $\msgpair{\KVAL}{k}$ after outputting $(k, 1)$ from $\GC_2$. So, after one round following $\GC_2$ (i.e.\ $7 \leq 6h(T) + 1$ rounds after $\TC(T)$ begins), the honest parties all learn $k$ and start running $\TC(H_k)$ with the common input $w_k$. Every honest party outputs $w_k$ from $\TC(H_k)$ once $7 + 6h(H_k) \leq 6h(T) + 1$ rounds have passed, and \mbox{thus outputs $\sigma$ or $w_k$ from $\TC(T)$ in at most $6h(T) + 1$ rounds.} \qedhere
        \end{itemize}
    \end{itemize}
\end{proof}

\subparagraph*{Complexity of $\TC(T)$.} The round complexity of $\TC(T)$ is at most $6h(T) + 1$, by Theorem \ref{tc-security}. If $T$'s maximum degree is $\Delta$, then in each of $\TC(T)$'s at most $h(T)$ recursive iterations the $\GC_2$ instance is run with inputs in $\{0, 1, \dots, \Delta\}$ and the $\KVAL$ messages carry values in $\{1, \dots, \Delta\}$. This means that each iteration costs $\BO(n^2)$ messages, each of size at most $\BO(\log \Delta + \log(h(T)))$, where the $\log (h(T))$ term is due to the iteration ID tags which distinguish different iterations' messages from each other. So, $\TC(T)$'s total message complexity is $\BO(n^2h(T))$, and its total communication complexity is $\BO(n^2h(T)(\log \Delta + \log(h(T))))$ bits. Finally, for when we need edge agreement in paths, note that $h(T) = q$ if $T$ is a path of length $2^q$ for any $q \geq 0$.

\subparagraph*{An Alternative Way.} A recent work has reduced edge agreement in a finite tree $T = (V, E)$ to two instances of edge agreement in paths of length $\BO(|V|)$ \cite{trees2025}. If one uses this reduction with our $\TC$ protocol serving as the path edge agreement protocol, then edge agreement in a finite tree $T = (V, E)$ costs $\BO(\log |V|)$ rounds, $\BO(n^2\log |V|)$ messages and $\BO(n^2\log |V| \log \log |V|)$ bits of communication. The disadvantage of using the reduction here is that it costs $\Theta(\log |V|)$ rounds even if $h(T) = o(\log |V|)$. For example, the spider tree $T_{k, \mathcal{M}}$ for multivalued $k$-graded consensus with the input domain $\mathcal{M}$ (the tree defined in Figure \ref{gradedfig}'s caption) has $|\mathcal{M}|k + 1$ vertices, which means that the reduction allows edge agreement in $T_{k, \mathcal{M}}$ in $\Theta(\log |\mathcal{M}| + \log k)$ rounds, while $\TC(T_{k, \mathcal{M}})$ only takes $\BO(\log k)$ rounds because $h(T_{k, \mathcal{M}}) = \BO(\log k)$.

\section{Edge Agreement in Infinite Paths} \label{sec_z}

In this section, we achieve edge agreement in the infinite paths $\mathbb{N}$ and $\mathbb{Z}$. First, we achieve edge agreement in $\mathbb{N}$ with a protocol based on exponential search where we again use $\GC_2$ to repeatedly shrink the agreement domain. Then, we build upon this protocol so that the parties need roughly half as many rounds to reach edge agreement when they have very large inputs. Finally, we upgrade edge agreement in $\mathbb{N}$ to edge \mbox{agreement $\mathbb{Z}$ with one initial $\GC_2$ iteration.}

\subsection{\texorpdfstring{Edge Agreement in $\mathbb{N}$}{Edge Agreement in N}} \label{sec_n}

We begin with a sequence of protocols $\Exp_0, \Exp_1, \dots$, with each protocol $\Exp_j$ allowing the parties to reach edge agreement in the infinite path $(2^j - 1, \dots)$. In $\Exp_j$, we use $\GC_2$ to reduce edge agreement in $(2^j - 1, \dots)$ to edge agreement in either the left path $(2^j - 1, \dots, 2^{j + 1} - 1)$, which the parties can reach via $\TC$, or the right path $(2^{j + 1} - 1, \dots)$, which the parties can recursively reach via $\Exp_{j+1}$. At some point, this recursion ends: When the parties run $\Exp_0$ with inputs in $\{0, \dots, M\}$ for some $M \geq 0$, there is eventually some $q \leq \floor{\log_2(\max(M, 1))}$ such that in $\Exp_{q}$ (which recursively appears in $\Exp_0$) the parties do not prefer the right path $(2^{q+1} - 1, \dots)$ as they have inputs below $2^{q+1}$, and consequently they reach edge agreement in the finite left path $(2^q - 1, \dots, 2^{q + 1} - 1)$.

\begin{algobox}[\textbf{ with the input }$v_i$]{Protocol $\Exp_j$}
    \State join a common instance of $\mathsf{GC}_2$
    \State let your $\GC_2$ input be $\LEFT$ if $v_i < 2^{j + 1}$, and let it be $\RIGHT$ otherwise\llabel{exp-gc-input}
    \State wait until you output some $(k, g)$ from $\mathsf{GC}_2$\llabel{exp-gc-wait}
    \State $\mathsf{side} \gets k$
    \If{$g = 1$}
        \State $v_i^{\sfnext} \gets 2^{j + 1} - 1$
        \State multicast $k$
    \ElsIf{$(k, g) = (\LEFT, 2)$}
        \State $v_i^{\sfnext} \gets \min(v_i, 2^{j + 1} - 1)$\llabel{exp-le-2-val}
    \ElsIf{$(k, g) = (\RIGHT, 2)$}
        \State $v_i^{\sfnext} \gets \max(v_i, 2^{j + 1} - 1)$\llabel{exp-ri-2-val}
    \ElsIf{$(k, g) = (\bot, 0)$}
        \State $v_i^{\sfnext} \gets 2^{j + 1} - 1$
        \State output $v_i^{\sfnext}$ from $\Exp_j$\llabel{exp-out-bot}
        \State multicast $\CENTER$\llabel{exp-center-cast}
        \parState{wait until you have received some $k \in \{\LEFT, \RIGHT\}$ from $t + 1$ parties, and let $\mathsf{side} \gets k$ when this happens}\llabel{exp-wait-side}
    \EndIf
    \If{$\mathsf{side} = \LEFT$}
        \parState{run an instance of $\TC((2^j - 1, \dots, 2^{j + 1} - 1))$ with the other parties for edge agreement in the path $(2^j - 1, \dots, 2^{j + 1} - 1)$, where your input is $v_i^{\sfnext}$}\llabel{exp-run-tc}
        \State if $g \geq 1$, then when you output $y$ from $\TC$, output $y$ from $\Exp_j$\llabel{exp-out-le}
    \ElsIf{$\mathsf{side} = \RIGHT$}
        \parState{run an instance of $\Exp_{j + 1}$ with the other parties for edge agreement in the path $(2^{j + 1} - 1, \dots)$, where your input is $v_i^{\sfnext}$}\llabel{exp-run-exp-next}
        \State if $g \geq 1$, then when you output $y$ from $\Exp_{j + 1}$, output $y$ from $\Exp_j$\llabel{exp-out-ri}
    \EndIf
    \Upon{receiving $\CENTER$ from $t + 1$ parties}
        \State output $2^{j + 1} - 1$ if you haven't output before and ignore future output commands\llabel{exp-center-out} \Comment{Note that $P_i$ should run Line \ref{line:exp-center-out} as soon as it receives $\CENTER$ from $t + 1$ parties, even if it is waiting for something else, e.g.\ waiting to output from $\GC_2$ on Line \ref{line:exp-gc-wait}}
    \EndUpon
\end{algobox}

\usetikzlibrary{positioning}
\usetikzlibrary{math}
\begin{figure}[ht]
    \centering
    \begin{tikzpicture}[scale=1]
        \tikzmath{\xval = 2.53;}
        \tikzmath{\yval = 10;}
        \node[] at (-2*\xval, 0) (l2) {$(\LEFT, 2)$};
        \node[] at (-\xval, 0) (l1) {$(\LEFT, 1)$};
        \node[] at (0, 0) (center) {$(\bot, 0)$};
        \node[] at (\xval, 0) (r1) {$(\RIGHT, 1)$};
        \node[] at (2*\xval, 0) (r2) {$(\RIGHT, 2)$};
        \draw[] (l2)--(l1)--(center)--(r1)--(r2);
        \node[below=0.5pt of l2,align=center,font=\fontsize{6.5pt}{\yval pt}\selectfont] {Run $\TC$ with the input\\$v_i^{\sfnext} = \min(v_i, 2^{j+1} - 1)$.\\Upon outputting $y$ from $\TC$,\\output $y$ from $\Exp_j$.\\Upon rcving $t + 1$ $\CENTER$ msgs,\\output $2^{j+1} - 1$ from $\Exp_j$.};
        \node[above=0.5pt of l1,align=center,font=\fontsize{6.5pt}{\yval pt}\selectfont] {Multicast $\LEFT$.\\Run $\TC$ with the input $2^{j+1} - 1$.\\Upon outputting $y$ from $\TC$,\\output $y$ from $\Exp_j$.\\Upon rcving $t + 1$ $\CENTER$ msgs,\\output $2^{j+1} - 1$ from $\Exp_j$.};
        \node[below=0.5pt of center,align=center,font=\fontsize{6.5pt}{\yval pt}\selectfont] {Multicast $\CENTER$.\\Output $2^{j+1} - 1$ from $\Exp_j$.\\If you receive $t + 1$ $\LEFT$ msgs,\\run $\TC$ with the input $2^{j+1} - 1$.\\If you receive $t + 1$ $\RIGHT$ msgs,\\run $\Exp_{j+1}$ with the input $2^{j+1} - 1$.\\Upon rcving $t + 1$ $\CENTER$ msgs,\\output $2^{j+1} - 1$ from $\Exp_j$.};
        \node[above=0.5pt of r1,align=center,font=\fontsize{6.5pt}{\yval pt}\selectfont] {Multicast $\RIGHT$.\\Run $\Exp_{j+1}$ with the input $2^{j+1} - 1$.\\Upon outputting $y$ from $\Exp_{j+1}$,\\output $y$ from $\Exp_j$.\\Upon rcving $t + 1$ $\CENTER$ msgs,\\output $2^{j+1} - 1$ from $\Exp_j$.};
        \node[below=0.5pt of r2,align=center,font=\fontsize{6.5pt}{\yval pt}\selectfont] {Run $\Exp_{j+1}$ with the input\\$v_i^{\sfnext} = \max(v_i, 2^{j+1} - 1)$.\\Upon outputting $y$ from $\Exp_{j+1}$,\\output $y$ from $\Exp_j$.\\Upon rcving $t + 1$ $\CENTER$ msgs,\\output $2^{j+1} - 1$ from $\Exp_j$.};
    \end{tikzpicture}
    \caption{A depiction of how each party behaves in $\Exp_j$, depending on the party's $\mathsf{\mathsf{GC}_2}$ output. The crucial observation is that the parties always obtain $\GC_2$ outputs that are adjacent in the figure. Note that a party can output $2^{j+1} - 1$ by receiving $t + 1$ $\CENTER$ messages before \mbox{it outputs from $\GC_2$.}}
\end{figure}

\apxthm{theorem}{expsecurity}{
    For any positive integer $q$ and any $0 \leq j < q$, suppose the honest parties run $\Exp_j$ with inputs in $\{2^j - 1, 2^q - 1\}$. Then, they reach edge agreement in the path $(2^j - 1, 2^q - 1)$ based on their inputs in at most $6(j+1)$ rounds if they have the common input $2^j - 1$ and in at most $7 + 12q - 6j$ rounds otherwise.
}

We prove Theorem \ref{expsecurity} in the \hyperref[expsecurity*]{appendix}, with a proof similar to the one of Theorem \ref{tc-security}.

\subparagraph*{Complexity of $\Exp_0$.} Consider an $\Exp_0$ execution where the maximum honest input is $M$. Let $q = \floor{\log_2(\max(M, 1))}$, i.e.\ let $q$ be the least non-negative integer such that every honest input is in $\{0, \dots, 2^{q+1} - 1\}$. In the execution, there is some minimum $k \in \{0, \dots, q\}$ such that in $\Exp_k$ (which is recursively a subprotocol of $\Exp_0$ if $k > 0$), the parties do not set $\mathsf{side} = \RIGHT$, and therefore do not execute $\Exp_{k+1}$. So, the $\Exp_0$ execution consists of the $\GC_2$ instances and the $\LEFT$/$\CENTER$/$\RIGHT$ multicasts of $\Exp_0, \Exp_1, \dots, \Exp_k$, plus the $\TC((2^k - 1,\dots,2^{k+1} - 1))$ instance inside $\Exp_k$. In total, these amount to $\BO(n^2k) = \BO(n^2q) = \BO(n^2\log M)$ messages, each of size $\BO(\log k) = \BO(\log q) = \BO(\log \log M)$ due to the message tags (since for all $j \geq 0$ one can assign $\BO(\log j)$-bit tags to $\Exp_j$ and to its $\GC_2$ and $\TC$ subprotocols). Meanwhile, the round complexity is at most \mbox{$12q + 19 = 12\log_2 M + \BO(1)$, by Theorem \ref{expsecurity}.}

\subsection{Nearly Halving the Round Complexity}

The protocol $\Exp_0$ is based on exponential search, which is a search algorithm to find a target value $x \in \mathbb{N}$ with approximately $2\log_2 x$ checks for different choices of $v$ whether $v \leq x$ or not. Each such check roughly corresponds to a 6-round $\GC_2$ instance in $\Exp_0$, and thus $\Exp_0$ takes $12\log_2 M + \BO(1)$ rounds by the simple fact that $2 \cdot 6 = 12$. However, it is possible to almost halve the number of checks. The strategy for this, from \cite{bentley76}, is to accelerate exponential search by using it to find $q = \floor{\log_2(x + 1)}$ instead of $x$, and then, knowing that $2^q - 1 \leq x \leq 2^{q + 1} - 2$, to \mbox{run binary search with the lower and upper bounds $2^q - 1$ and $2^{q+1} - 2$ to find $x$.}

Below, we use an analogous strategy in our protocol $\TStep(\Agn)$ for edge agreement in $\mathbb{N}$, \linebreak where we let $\Agn$ be any protocol for edge agreement in $\mathbb{N}$, and accelerate $\Agn$ by using it in a two-stage manner. In $\TStep(\Agn)$, the parties roughly speaking use $\Agn$ to agree on some $k$ such that the path $(2^k - 1, \dots, 2^{k + 1} - 1)$ intersects the convex hull of the honest inputs, with each party $P_i$ with the input $v_i$ wanting $k$ to be $\floor{\log_2(v_i + 1)}$. Then, each party with an input outside this path adopts a new safe value in the path (the endpoint of the path closest to its value), and finally the parties run $\TC$ to reach edge agreement in the path. Naturally, we face the issue that $\Agn$ does not guarantee agreement on a single $k$, which means that the parties might output different (though adjacent) $k$ values from it and thus behave differently after outputting from $\Agn$. Luckily, by making each party $P_i$ let its $\Agn$ input be what it wants $5k$ to be ($5\floor{\log_2(v_i + 1)}$) rather than what it wants $k$ to be, we can ensure that the discrepancy \mbox{between how each party behaves after outputting from $\Agn$ is sufficiently small.}

\begin{algobox}[\textbf{ with the input }$v_i$]{Protocol $\TStep(\Agn)$}
    \State join a common instance of $\Agn$ where your input is $5\floor{\log_2(v_i + 1)}$
    \Upon{outputting $5k_i + r_i$ from $\Agn$ for some $r_i \in \{0, \dots 4\}$}
        \If{$r_i = 0$}
            \State $v_i^{\sfnext} \gets \min(\max(v_i, 2^{k_i} - 1), 2^{k_i + 1} - 1)$
        \Else
            \State $v_i^{\sfnext} \gets 2^{k_i + 1} - 1$
        \EndIf
        \If{$0 \leq r_i \leq 2$}
            \parState{run an instance of $\TC((2^{k_i} - 1, \dots, 2^{k_i + 1} - 1))$ with the other parties for edge agreement in $(2^{k_i} - 1, \dots, 2^{k_i + 1} - 1)$, where your input is $v_i^{\sfnext}$}
            \State if $r_i \leq 1$, then when you output $y$ from $\TC$, output $y$ from $\TStep(\Agn)$
        \EndIf
        \If{$2 \leq r_i \leq 3$}
            \State output $v_i^{\sfnext}$
        \EndIf
        \If{$3 \leq r_i \leq 4$}
            \parState{run an instance of $\TC((2^{k_i + 1} - 1, \dots, 2^{k_i + 2} - 1))$ with the other parties for edge agreement in $(2^{k_i + 1} - 1, \dots, 2^{k_i + 2} - 1)$, where your input is $v_i^{\sfnext}$}
            \State if $r_i = 4$, then when you output $y$ from $\TC$, output $y$ from $\TStep(\Agn)$
        \EndIf
    \EndUpon
\end{algobox}

\apxthm{theorem}{dexpsecurity}{
    Suppose $\Agn$ is a live edge agreement protocol for $\mathbb{N}$ which for some function $f$ takes at most $f(M)$ rounds when the honest parties run it with inputs in $\{0, \dots, M\}$. Then, $\TStep(\Agn)$ is a live edge agreement protocol for $\mathbb{N}$ that takes at most $f(5\floor{\log_2(M + 1)}) + 6\floor{\log_2(M + 1)}+1$ rounds when the honest parties run it with inputs in $\{0, \dots, M\}$.
}

We prove Theorem \ref{dexpsecurity} in the \hyperref[dexpsecurity*]{appendix}, with a proof similar to the one of Theorem \ref{tc-security}. Although $\TStep(\Agn)$ does not use graded consensus, it depends on the same core idea as $\TC$ and $\Exp$: The parties reach graded consensus on how to behave, and thus they behave well together. We do not need an actual graded consensus protocol in $\TStep(\Agn)$ to implement this idea as the ``grades'' in $\TStep(\Agn)$ are provided directly by $\Agn$, in the form of the $r_i$ value which a party $P_i$ obtains \mbox{as the remainder when it divides its $\Agn$ output by $5$.}

Note that $\TStep(\Agn)$'s round complexity $f(5\floor{\log_2(M + 1)}) + 6\floor{\log_2(M + 1)}+1$ stated in Theorem \ref{dexpsecurity} follows straightforwardly from $\Agn$ taking at most $f(5\floor{\log_2(M + 1)})$ rounds and $\TC$ taking at most $6\floor{\log_2(M + 1)}+1$ rounds.

Observe that $\TStep$ allows us to obtain an infinite sequence of protocols $\Pi_1, \Pi_2, \Pi_3, \dots$ for edge agreement in $\mathbb{N}$ where $\Pi_1 = \Exp_0$ and $\Pi_k = \TStep(\Pi_{k-1})$ for all $k \geq 2$, with each protocol $\Pi_k$ corresponding to the unbounded search algorithm $B_k$ in \cite{bentley76}. Letting $f_k(M)$ for each $k \geq 1$ be $\Pi_k$'s round complexity depending on the maximum honest input $M$, we have $f_1(M) = 12\log_2 M + \BO(1)$ and have $f_{k+1}(M) = f_k(5\floor{\log_2(M + 1)}) + 6\floor{\log_2(M + 1)} + 1$ for all $k \geq 1$. One can show by induction that $f_k(M) = 6(L^{(k)}(M) +\sum_{j = 1}^{k} L^{(k)}(M)) + \BO(k) $ for all $k \geq 1$, where $L^{(1)}(v) = \floor{\log_2(v + 1)}$ and $L^{(j)}(v) = L(L^{(j - 1)}(v))$ for all $j \geq 2$.

The protocol $\Pi_2 = \TStep(\Exp_0)$ above notably takes $\frac{1}{2} + \BO(\frac{\log \log M}{\log M}) = \frac{1}{2} + o(1)$ times as many rounds as $\Pi_1 = \Exp_0$, and only $1 + o(1)$ times as many rounds as the protocols $\Pi_3, \Pi_4, \dots$ beyond it. \mbox{Therefore, in the rest of the paper we use $\TStep(\Exp_0)$ for edge agreement in $\mathbb{N}$.}

\subparagraph*{Complexity of $\TStep(\Exp_0)$.} Consider a $\TStep(\Exp_0)$ execution where the maximum honest input is $M$, and let $q = \floor{\log_2 (M + 1)}$. The initial $\Exp_0$ in $\TStep(\Exp_0)$ takes $12\log_2 q + \BO(1)$ rounds, with $\BO(n^2\log q)$ messages and $\BO(n^2\log q \log \log q)$ bits of communication in total. Afterwards, in $\TStep(\Exp_0)$'s remaining $6q+1$ rounds, each party $P_i$ runs an instance of $\TC$, where $z_i$ is the party's $\Exp_0$ output that is upper bounded by $5q$. Due to $\GC_2$ and therefore $\TC$ being fully symmetric protocols with balanced communication, each party $P_i$ sends $\BO(nz_i) = \BO(nq)$ messages and $\BO(nz_i\log z_i) = \BO(nq\log q)$ bits to the others in $\TC$. Hence, we conclude that $\TStep(\Exp_0)$ takes $12\log_2 q + 6q + \BO(1) = 6\log_2 M + 12\log_2\log_2 M + \BO(1)$ rounds, with the total message complexity $\BO(n^2q) = \BO(n^2\log M)$ and the total communication complexity of $\BO(n^2q\log q) = \BO(n^2\log M \log \log M)$ bits.

\subsection{\texorpdfstring{Edge Agreement in $\mathbb{Z}$}{Edge Agreement in Z}}

As previously, let $\Agn$ be a protocol for edge agreement in $\mathbb{N}$. Observe that the parties can reach edge agreement in $(\dots, -1, 0)$ via $\Agn$ by mirroring their inputs (multiplying them by $-1$), running $\Agn$ with their mirrored inputs, and mirroring their $\Agn$ outputs. This allows the protocol $\Agz(\Agn)$ below for edge agreement in $\mathbb{Z}$, where the parties run $\GC_2$ with the input $-1$ if they have negative $\Agz$ inputs and the input $1$ otherwise, and then use their $\GC_2$ outputs to decide on both their $\Agn$ inputs and on if they want to run $\Agn$ with mirroring or not.  Note that \cite{glw24} achieves convex agreement in $\mathbb{Z}$ similarly, though there the parties reach byzantine \mbox{agreement instead of 2-graded consensus on if they have negative inputs or not.}

\begin{algobox}[\textbf{ with the input }$v_i$]{Protocol $\Agz(\Agn)$}
    \State join a common instance of $\mathsf{GC}_2$ and a common instance of $\Agn$
    \State let your $\GC_2$ input be $1$ if $v_i \geq 0$, and let it be $-1$ otherwise
    \State wait until you output some $(k, g)$ from $\mathsf{GC}_2$
    \If{$(k, g) = (\bot, 0)$}
        \State let your $\Agn$ input be $0$
        \State output $0$ from $\Agz(\Agn)$
    \Else
        \State let your $\Agn$ input be $v_i^{\sfnext} = \max(0, (g - 1) \cdot k \cdot v_i)$
        \State when you output $y$ from $\Agn$, output $k \cdot y$ from $\Agz(\Agn)$
    \EndIf
\end{algobox}

\apxthm{theorem}{agntoagz}{
    If $\Agn$ is a live protocol for edge agreement in $\mathbb{N}$, then $\Agz(\Agn)$ is a live protocol for edge agreement in $\mathbb{Z}$.
}

We prove Theorem \ref{agntoagz} in the \hyperref[agntoagz*]{appendix}, with a proof similar to the one of Theorem \ref{tc-security}.

\subparagraph{Complexity.} When the honest parties run $\Agz(\Agn)$ with inputs in $\{-M, \dots, M\}$ for some $M \geq 0$, they run $\GC_2$, and then run $\Agn$ with inputs upper bounded by $M$. So, the $\Agz$ outer shell adds an overhead of $6$ rounds, $\BO(n^2)$ messages and $\BO(n^2)$ bits of communication to $\Agn$. If $\Agn = \TStep(\Exp_0)$, then $\Agz(\Agn)$'s complexity is $6\log_2 M + \BO(\log \log M)$ rounds, $\BO(n^2\log M)$ messages and $\BO(n^2\log M \log \log M)$ bits of communication.

\section{Termination} \label{sec_termination}

The edge agreement and graded consensus protocols we present in the other sections achieve liveness, but they do not terminate. That is, they require every honest party to run forever, and lose liveness if any honest party ever halts. Fortunately, there exists a simple 3-round quadratic-complexity protocol $\Term$ which can address this shortcoming. We present it in the \hyperref[termAppendix]{appendix} with a proof of its security. It is the same termination procedure as the one in \cite{qrbc25} (which terminates edge agreement in paths of the form $(0, \frac{1}{k}, \frac{2}{k}, \dots, 1$), with only surface-level changes. It relies on the fact that edge agreement protocols guarantee agreement on at most two values; that is, they guarantee that for some $y$ and $y'$ everybody outputs $y$ or $y'$. Note that unlike our other protocols, $\Term$ does not require every \mbox{honest party to acquire an input.}

\apxthm{theorem}{termSec}{
    In any $\Term$ execution where the maximum honest message delay is $\Delta$ and the honest parties only acquire inputs in some set $\{y, y'\}$, if by some time $T$ either every honest party acquires an input in $\{y, y'\}$ or some honest party terminates $\Term$, then by the time $T + 3\Delta$ every honest party terminates $\Term$ \mbox{with an output that is an honest party's input.}
}

Since $\Term$ is not novel, we present it and prove Theorem \ref{termSec} in the \hyperref[termAppendix]{appendix}. Our proof is similar to the one in \cite{qrbc25}, though \cite{qrbc25} does not prove that $\Term$ takes at most 3 rounds.

For any live edge agreement protocol $\Pi$, there is a sequentially composed protocol $\Term \circ \Pi$ where the parties run $\Pi$ to reach edge agreement, and run $\Term$ to terminate. In $\Term \circ \Pi$, each party lets its $\Pi$ input be its edge agreement inputs, lets its $\Term$ input be its $\Pi$ output (if and when it outputs from $\Pi)$, and terminates with the output $y$ upon terminating $\Term$ with any output $y$. The composed protocol $\Term \circ \Pi$ inherits the edge agreement and convex validity properties of $\Pi$ since each honest party's $\Term$ output is some honest party's $\Pi$ output. Furthermore, if $\Pi$ achieves liveness in $k$ rounds, then $\Term \circ \Pi$ terminates in $k + 3$ rounds. This is because after $k$ rounds, either every honest has output from $\Pi$ and thus acquired a $\Term$ input, or some honest party has terminated $\Term$. In either case, the honest parties all terminate $\Term$ within $3$ more rounds, i.e.\ within $k + 3$ rounds of every honest party acquiring a $\Pi$ input. Note that some honest parties might terminate $\Term$ before $k$ rounds have passed and stop running $\Pi$ before everybody outputs from $\Pi$. This might cause $\Pi$ to lose its liveness, meaning that some honest parties might never output from $\Pi$ and thus never acquire $\Term$ inputs. This is not an issue because if some honest party terminates $\Term$, then every honest party terminates $\Term$, \mbox{even if not every honest party acquires a $\Term$ input.}

In $\Term$ (which we present in the \hyperref[termAppendix]{appendix}), each honest party multicasts one constant-size $\READY$ message, and multicasts at most one $\msgpair{\ECHO}{v}$ message for each value $v$ that is an honest party's $\Term$ input. So, for any edge agreement protocol $\Pi$, the overhead in $\Term \circ \Pi$ due to $\Term$ is $3$ rounds and $\BO(n^2)$ messages, each at most \mbox{roughly the size of an honest $\Pi$ output.}

\subparagraph*{Termination for edge agreement in $\mathbb{Z}$.} If the parties terminate with $\Term$ after reaching edge agreement in $\mathbb{Z}$ with inputs of magnitude at most $M$, then each honest $\Term$ input is an integer in $\{-M, \dots, M\}$, with a $\BO(\log M)$-bit representation. So, for any live protocol $\Agz$ for edge agreement in $\mathbb{Z}$, $\Term \circ \Agz$ is a terminating protocol for edge agreement in $\mathbb{Z}$ which costs $3$ more rounds, $\BO(n^2)$ more messages and $\BO(n^2\log M)$ \mbox{more bits of communication than $\Agz$.}

\subparagraph*{Termination for edge agreement in trees.} If the parties run $\Term$ to terminate after they reach edge agreement in a tree $T = (V, E)$, then each honest $\Term$ input is a vertex in $V$, with a $\BO(\log|V|)$-bit representation. So, $\Term \circ \TC(T)$ is a terminating protocol for edge agreement in $T$ which costs $3$ more rounds, $\BO(n^2)$ more messages and $\BO(n^2\log |V|)$ more bits of communication than $\TC(T)$. Note that $\Term$ does not cause an asymptotic communication complexity increase here because $\log_2|V| = \BO(h(T)\log\Delta)$, where $\Delta$ is $T$'s maximum degree. This is because $h(T) = \Omega(\log_\Delta|V|)$. To see why, let $h_\Delta(k)$ for any $k \geq 1$ be the minimum value $h(T)$ can have for any tree $T$ of maximum degree at most $\Delta \geq 2$ with $k$ vertices, and observe that $h_\Delta(k) \geq 1 + h_\Delta(\ceil{\frac{k - 1}{\Delta}})$ for all $k \geq 2$ since removing any centroid of degree $d \in [\Delta]$ from a tree with $k \geq 2$ vertices must create a forest whose largest component has at least $\ceil{\frac{k - 1}{\Delta}}$ vertices. By induction, \mbox{this implies $h_{\Delta}(\Delta^k) \geq k$ for all $k \geq 0$, which means that $h(T) = \Omega(\log_\Delta|V|)$.}

\section{Extension to Real Numbers} \label{sec_reduction}

Termination is harder for approximate agreement in $\mathbb{R}$. As approximate agreement in $\mathbb{R}$ does not entail agreement on two values, $\Term$ alone cannot directly provide termination. So, to achieve approximate agreement in $\mathbb{R}$ with termination, we reduce \mbox{it to edge agreement in $\mathbb{Z}$.}

\begin{theorem}
    For all $\varepsilon > 0$, $\varepsilon$-agreement in $\mathbb{R}$ with the maximum honest input magnitude $M$ can be reduced to edge agreement in $\mathbb{Z}$ with the maximum honest input magnitude $\ceil{\frac{2M}{\varepsilon} - \frac{1}{2}}$.
\end{theorem}

\begin{proof}
    We reduce $2$-agreement in $\mathbb{R}$ with the maximum honest input magnitude $M$ to edge agreement in $\mathbb{Z}$ with the maximum honest input magnitude $\ceil{M - \frac{1}{2}}$. This reduction implies the theorem since the parties can reach $\varepsilon$-agreement in $\mathbb{R}$ by multiplying their inputs by $\frac{2}{\varepsilon}$, reaching $2$-agreement, and dividing their outputs by $\frac{2}{\varepsilon}$.

    Each party $P_i$ rounds its $2$-agreement input $v_i$ to the nearest integer $v_i'$, rounding towards $0$ in the case of ties. Then, the parties reach edge agreement in $\mathbb{Z}$ with their rounded integers of magnitude at most $\ceil{M - \frac{1}{2}}$, and each party $P_i$ obtains some edge agreement output $y_i'$. For all honest parties $P_i$ and $P_j$ it holds that $|y_i' - y_j'| \leq 1$ and that $v_a' \leq y_i' \leq v_b'$, where $a$ and $b$ are respectively the indices of the honest parties with the minimum and maximum edge \linebreak agreement inputs. Observe that this implies $v_a - \frac{1}{2} \leq v_a' \leq y_i' \leq v_b' \leq v_b + \frac{1}{2}$ for all honest $P_i$.

    Afterwards, each party $P_i$ uses its 2-agreement input $v_i$ to convert $y_i'$ into its  $2$-agreement output $y_i$ by letting $y_i = \min(y_i' + \frac{1}{2}, v_i)$ if $y_i' \leq v_i$, and $y_i = \max(y_i' - \frac{1}{2}, v_i)$ otherwise. That is, $P_i$ obtains $y_i$ by moving from $y_i'$ towards its input $v_i$ by a distance of at most $\frac{1}{2}$. For all honest $P_i$ and $P_j$ it holds that $\min(v_a, v_i) \leq y_i \leq \max(v_b, v_i)$, which implies validity, and that $|y_i - y_j| \leq |y_i - y_i'| + |y_i' - y_j'| + |y_j' - y_j| \leq \frac{1}{2} + 1 + \frac{1}{2} = 2$, which implies $2$-agreement. \qedhere
\end{proof}

When $M$ is the maximum honest input magnitude, since we can achieve edge agreement in $\mathbb{Z}$ in $6\log_2 M + \BO(\log\log M)$ rounds with $\BO(n^2\log M)$ messages and $\BO(n^2\log M \log \log M)$ bits of communication, the reduction enables $\varepsilon$-agreement in $\mathbb{R}$ in $6\log_2 \frac{M}{\varepsilon} + \BO(\log \log \frac{M}{\varepsilon})$ rounds with $\BO(n^2\log \frac{M}{\varepsilon})$ messages and $\BO(n^2\log \frac{M}{\varepsilon}\log\log \frac{M}{\varepsilon})$ bits of communication.

Note that our reduction allows the parties to reach $\varepsilon$-agreement in $\mathbb{R}$ by only sending each other discrete messages. So, it allows approximate agreement without any rounding errors. In contrast, protocols where the parties send values in $\mathbb{R}$ \cite{dolev86, aad04, glw22, delphi} must in order to avoid rounding errors assume that the inputs in $\mathbb{R}$ fit $\ell$-bit strings, with $\ell$ affecting complexity.

\subparagraph{Future Work.} It remains open to design an asynchronous quadratic-communication protocol for $\varepsilon$-agreement in $\mathbb{R}$ that tolerates $t < \frac{n}{3}$ faults in $\BO(\log \frac{S}{\varepsilon})$ rounds, where $S$ is the honest input spread. Our protocol takes $\BO(\log \frac{M}{\varepsilon})$ rounds, where possibly $M \gg S$, while Abraham et al.\ achieve the lower round complexity $\BO(\log \frac{S}{\varepsilon})$ but with cubic communication \cite{aad04}. In fact, this task is open not just for asynchronous networks, but for synchronous networks as well. The classical synchronous protocol in \cite{dolev86} which at first seems to fit the bill in fact takes as many rounds as the adversary desires because its round complexity scales with the spread of \emph{all} inputs, including fake byzantine ones. \mbox{It would be a good first step to solve this issue.}

\bibliography{refs}

\appendix

\section{Graded Consensus Protocols} \label{GradedConsensus}

In this section, we construct a family of $2^k$-graded consensus protocols $\mathsf{GC}_{2^0}, \mathsf{GC}_{2^1}, \mathsf{GC}_{2^2}, \dots$ that support inputs in $\{0, 1\}^\ell$ for any public length parameter $\ell \geq 1$. The complexity of each protocol $\GC_{2^k}$ is $3k + 3$ rounds, $\BO(kn^2)$ messages and $\BO(kn^2(\ell + \log k))$ bits of communication. To obtain these protocols we begin with $\mathsf{GC}_1$, and then for all $k \geq 1$ obtain $\GC_{2^k}$ by grade-doubling $\mathsf{GC}_{2^{k-1}}$ in 3 rounds. While we only use $\GC_2$ in this paper, note that grades beyond 2 \linebreak can also be useful, as $2^k$-graded consensus forms the ``expand'' part of the ``expand-and-extract'' paradigm of achieving byzantine \mbox{agreement with an $\BO(2^{-k})$ error probability \cite{fmll21}.}

As $2^k$-graded consensus is a special case of edge agreement with the centroid decomposition height $k + 1$ (see Figure \ref{gradedfig}), our edge agreement protocol $\TC$ is also a $2^k$-graded consensus protocol, but one that takes at most $6k + 7$ rounds. The more restricted definition of $2^k$-graded consensus allows us to achieve it in $3k + 3$ rounds instead. This is the round complexity of the binary $2^k$-graded consensus protocol in \cite{delphi}, and it is for all $k \geq 0$ lower than the round complexity of any previous multivalued $2^k$-graded consensus protocol which we are aware of that tolerates $t < \frac{n}{3}$ faults with perfect security. Note that $2^k$-graded consensus is possible in less rounds than $3k + 3$ if $t < \frac{n}{5}$ ($k + 1$ rounds \cite{dolev86, aw24}) or if the parties can use a cryptographic setup ($k + 2$ rounds, \cite{simon25}). If in our edge agreement protocols we replaced $\GC_2$ with a faster $r$-round $2$-graded consensus protocol, then our edge protocols would accordingly become $\frac{6}{r}$ times faster, excluding some \mbox{constant round complexity terms that would not change.}

\subsection{1-Graded Consensus}

Below is our live 3-round multivalued 1-graded consensus protocol $\GC_1$.

\begin{algobox}[ with the input $v_i$]{Protocol $\mathsf{GC}_1$}
    \State $V_i^1, V_i^2, \dots, V_i^\ell \gets \varnothing, \varnothing, \dots, \varnothing$
    \State $W_i^1, W_i^2, \dots, W_i^\ell \gets \varnothing, \varnothing, \dots, \varnothing$
    \State multicast $\msgpair{\ECHO}{v_i}$

    \Upon{receiving echoes ($\ECHO$ messages) on $\bot$ or on strings other than $v_i$ from $t + 1$ parties}
        \State multicast $\msgpair{\ECHO}{\bot}$
        \State output $(\bot, 0)$ if you haven't output before, and keep running
    \EndUpon

    \For{any bit $b \in \{0, 1\}$ and any index $k \in \{1, \dots, \ell\}$, when you have received echoes on $\bot$ or on strings with the $k^{\text{th}}$ bit $b$ from $t+1$ parties}\llabel{gc1v}
        \State $V_i^k \gets V_i^k \cup \{b\}$
        \If{$V_i^k = \{0, 1\}$}
            \State output $(\bot, 0)$ if you haven't output before, and keep running\llabel{gc1-vout}
        \EndIf
    \EndFor

    \For{any bit $b \in \{0, 1\}$ and any index $k \in \{1, \dots, \ell\}$, when you have received echoes on $\bot$ or on strings with the $k^{\text{th}}$ bit $b$ from $2t + 1$ parties}\llabel{gc1w}
        \State $W_i^k \gets W_i^k \cup \{b\}$
        \If{$W_i^1, W_i^2, \dots, W_i^\ell = \{c_1\}, \{c_2\}, \dots, \{c_\ell\}$ for some $c_1,\dots,c_\ell \in \{0,1\}^\ell$}\llabel{gc1cond}
            \State multicast $\msgpair{\PROP}{c_1 \| c_2 \| \dots \| c_\ell}$ \Comment{Here, $||$ stands for string concatenation.}
        \EndIf
    \EndFor

    \Upon{receiving ($\PROP$ messages) on some string $v$ from $n - t$ parties}
        \State let $y = (v, 1)$ if $v = v_i$, and let $y = (\bot, 0)$ otherwise
        \State output $y$ if you haven't output before, and keep running
    \EndUpon
\end{algobox}

This protocol is based on the 4-round protocol $\mathsf{AWC}$ in \cite{dm24}. Our addition is the list of sets $V_i^1, \dots, V_i^\ell$, which $\mathsf{AWC}$ does not have. While according to $\mathsf{AWC}$ a party $P_i$ would output $(\bot, 0)$ upon observing that $W_i^k = \{0, 1\}$ for some $k$, in $\GC_1$ a party $P_i$ outputs $(\bot, 0)$ upon observing that $V_i^k = \{0, 1\}$. As the parties fill their $V$ sets more readily than their $W$ sets (due to Line \ref{line:gc1v} having a weaker condition than Line \ref{line:gc1w}), the $V$ sets let us shave off one round.

\begin{theorem}
    $\mathsf{GC}_1$ is a secure $3$-round $1$-graded consensus protocol with liveness.
\end{theorem}

\begin{proof}
    The simplest property is intrusion tolerance. It follows from the fact that if an honest party outputs $(v, 1)$ for any $v$, then $v$ is the party's input, which makes it an honest input.

    Agreement is also simple. Suppose there are two honest parties $P$ and $P'$ who respectively output $(v, 1)$ and $(v', 1)$. Then, $P$ has received $n - t$ proposals on $v$, while $P'$ has received $n - t$ proposals on $v'$. Hence, there are $n - 2t \geq t + 1$ parties, or at least one honest party, who have sent $P$ a proposal on $v$ and $P'$ a proposal on $v'$. This implies that $v = v'$ because an honest party can only propose one value. The reason why this is true is that an honest party $P_i$ only proposes a value when it adds a bit to one of its $W$ sets $W_i^1, \dots, W_i^\ell$ and after doing so observes that $|W_i^1| = \dots = |W_i^\ell| = 1$, which can happen only once.

    For validity, suppose the honest parties have a common input $v^* = b_1^* \| \dots \| b_\ell^* \in \{0,1\}^\ell$. In this case, they echo $v^*$ (multicast $\msgpair{\ECHO}{v^*}$), and they do not echo $\bot$ since they do not receive $t + 1$ echoes on values other than $v^*$. For all $k \in \{1, \dots, \ell\}$, since no honest party echoes $\bot$ or a string with the $k^{\text{th}}$ bit $1 - b_k^*$, the bit $b_k^*$ is the only bit that any honest party $P_i$ can add to $V_i^k$ or $W_i^k$. Consequently, an honest party can neither output $(\bot, 0)$ on Line \ref{line:gc1-vout} nor propose anything other than $v^*$. This means that an honest party can only output after receiving $n - t$ proposals on $v^*$, which leads to \mbox{the party outputting $v^*$ since it has the input $v^*$.}

    Finally, for liveness in 3 rounds, suppose that by some time $T$ every honest party acquires its input and thus echoes it. We show via case analysis that every honest party outputs by the time $T + 3\Delta$, where $\Delta$ is the maximum honest \mbox{message delay that the adversary causes.} \begin{enumerate}
        \item The easier case is if there is no input that $t + 1$ honest parties have. Then, by the time $T + \Delta$, every honest party receives $n - 2t \geq t + 1$ echo messages on inputs that are not its own, and therefore becomes able to output $\bot$.
        \item Suppose instead that some $v^* = b_1^* \| \dots \| b_\ell^* \in \{0,1\}^\ell$ is the input of $t + 1$ honest parties. By the time $T + \Delta$, all the honest parties with inputs other than $v^*$ receive $t + 1$ echoes on $v^*$ and thus echo $\bot$. Then, by the time $T + 2\Delta$, for each $k \in \{1, \dots, \ell\}$ each honest party $P_i$ receives $n - t \geq 2t + 1$ echoes that are on $\bot$ or $v^*$ (which has the $k^{\text{th}}$ bit $b_i^*$), and thus adds $b_i^*$ to $W_i^k$. So, by the time $T + 2\Delta$, either every honest party $P_i$ observes that $W_i^k = \{b_k^*\}$ for all $k$, or some honest party $P_i$ obtains $W_i^k = \{0, 1\}$. In the former case, the honest parties all propose $v^*$, which by the time $T + 3\Delta$ leads to them all receiving $n - t$ proposals on $v^*$ and thus becoming able to output $(v^*, 1)$ or $(\bot, 0)$. In the latter case, the party $P_i$ that has obtained $W_i^k = \{0, 1\}$ has by the time $T + 2\Delta$ for both bits $b \in \{0, 1\}$ received echo messages on $\bot$ or on strings with the $k^{\text{th}}$ bit $b$ from $n - t$ parties ($t + 1$ or more of which are honest). Hence, by the time $T + 3\Delta$, for both bits $b \in \{0, 1\}$ every honest party $P_j$ receives from $t + 1$ honest parties echoes on $\bot$ or on strings with the $k^{\text{th}}$ bit $b$, which means that $P_j$ obtains $V_j^k = \{0, 1\}$ \mbox{and becomes able to output $(\bot, 0)$.} \qedhere
    \end{enumerate}
\end{proof}

\subparagraph{Complexity.} The round complexity is $3$, as proven above. The message complexity is $\BO(n^2)$ as each party multicasts at most $3$ messages (an echo on $\bot$, an echo on its input and a proposal), \linebreak and the communication complexity is $\BO(\ell n^2)$ bits. Note that the communication is balanced evenly between the parties: Each party sends $\BO(n)$ messages and $\BO(\ell n)$ bits to the others.

\subsection{Grade Doubling}

We achieve grade-doubling with a \emph{proposal} protocol we call $\Prop$. In proposal, each party $P_i$ acquires an input $v_i \in \mathcal{M}$ in an input domain $\mathcal{M}$, and outputs a set $Y_i \subseteq M$ of size $1$ or $2$. To work properly, a proposal protocol requires there to exist a set $S$ of size $2$ (which the parties might not know in advance) such that the honest parties only acquire inputs in $S$. If there exists such a set $S$, a proposal protocol \mbox{achieves liveness with the following safety properties:}

\begin{itemize}
    \item \textbf{\emph{agreement:}} If $Y_i$ and $Y_j$ are honest output sets, then $Y_i \cap Y_j \neq \varnothing$.
    \item \textbf{\emph{validity:}} For every honest output set $Y_i$, every $y \in Y_i$ is an honest party's input.
\end{itemize}

With binary inputs, proposal is equivalent to binary 1-graded consensus. Suppose the parties run a proposal protocol with inputs in $\{0, 1\}$. If they run it with a common input bit $b^*$ then they all output $\{b^*\}$ from it, and otherwise for some bit $b^*$ they each output $\{b^*\}$ or $\{0, 1\}$ from it. These output guarantees are exactly those of binary 1-graded consensus when we map the proposal outputs $\{0\}, \{0, 1\}, \{1\}$ to the 1-graded consensus outputs $(0, 1), (\bot, 0), (1, 1)$. However, when there are more than two possible inputs, a binary 1-graded consensus protocol cannot be used, whereas a proposal protocol still works with output guarantees comparable to those of 1-graded consensus as long as the honest parties happen to run it with at most two inputs $a$ and $b$, even if the parties with the input $a$ do not know what $b$ is and vice versa.

Our 3-round proposal protocol $\Prop$ is based on the 4-round protocol $\mathsf{AProp}$ in \cite{dm24}, which is in turn based on the protocol $\Pi_{\mathsf{prop}}^{t_s}$ in \cite{bkl19}. Again, we shave off one round. While according to $\mathsf{AProp}$'s design a party $P_i$ would need to receive $2t + 1$ echoes on a value $v$ to add $v$ to its set $V_i$, below in $\Prop$ the party $P_i$ adds $v$ to $V_i$ after receiving $t + 1$ echoes on $v$. This change makes it easier for $P_i$ to add values to $V_i$, obtain $|V_i| = 2$ and output $V_i$.

Note that using proposal to double grades is not a new invention. The BinAA protocol in \cite{delphi} essentially uses proposal (therein called weak binary value broadcast) for this purpose.

\begin{algobox}[\textbf{ with the input }$v_i$]{Protocol $\mathsf{Prop}$}
    \State multicast $\msgpair{\ECHO}{v_i}$ if you haven't done so before
    \For{any $v$, when you have received echoes on $v$ from $t + 1$ parties}
        \State multicast $\msgpair{\ECHO}{v}$ if you haven't done so before
        \State $V_i \gets V_i \cup \{v\}$
        \If{$|V_i| = 2$}
            \State output $V_i$ if you haven't output before, and keep running
        \EndIf
    \EndFor
    \When{there is a first $v$ such that you've received echoes on $v$ from $2t + 1$ parties \nolinebreak}
        \State multicast $\msgpair{\PROP}{v}$
    \EndWhen
    \Upon{receiving proposals on any $v$ from $n - t$ parties}
        \State output $\{v\}$ if you haven't output before, and keep running
    \EndUpon
\end{algobox}

\begin{theorem}
    $\mathsf{Prop}$ is a secure 3-round proposal protocol with liveness.
\end{theorem}
\begin{proof}
Below, we assume that there exists a set $S = \{a, b\}$ such that the honest parties only acquire inputs in $S$. Without this, the agreement and liveness properties would be lost.

If an honest party echoes a value $v$, then $v$ must be some honest party's input. Otherwise, the first honest party that echoes $v$ would have to have received $t + 1$ echoes on $v$, at least one being from an honest party who contradictorily echoed $v$ earlier. For any value $v$ that is not any honest party's input, the fact that no honest party echoes $v$ means that the honest parties cannot output sets that contain $v$ \mbox{(i.e.\ that we have validity), due to the following:}\begin{enumerate}
    \item An honest party $P_i$ cannot add $v$ to its set $V_i$, as $P_i$ would need to receive $t + 1$ echoes on $v$ to do this. Therefore, if $P_i$ outputs $V_i$ when $|V_i| = 2$, then $v \not \in V_i$.
    \item An honest party $P_i$ cannot propose $v$, as $P_i$ would need to receive $2t + 1$ echoes on $v$ to do this. This means that an honest party $P_j$ cannot output $\{v\}$, as $P_j$ would have to receive $n - t \geq t + 1$ proposals on $v$ to do this.
\end{enumerate}

As for agreement, observe that by the validity property the honest parties can only output non-empty subsets of $S = \{a, b\}$, as their inputs are all in $S$. So, we only need to show that if some honest parties $P$ and $P'$ respectively output $\{v\}$ and $\{v'\}$ for some $v, v' \in S$, then $v = v'$. Again, this follows from a traditional quorum intersection argument. If we have such parties $P$ and $P'$, then $P$ has received $n - t$ proposals on $v$, while $P'$ has received $n - t$ proposals on $v'$. So, there are $n - 2t \geq t + 1$ parties, or at least one honest party, who have sent $P$ a proposal on $v$ and $P'$ a proposal on $v'$. \mbox{This implies $v = v'$ as an honest party can only propose once.}

It remains to show liveness in 3 rounds. As we did for $\GC_1$, let us suppose that by some time $T$ every honest party has acquired its input $\{a, b\} \in S$ and thus echoed it. Let us show that every honest party outputs by the time $T + 3\Delta$, where $\Delta$ is the maximum honest message delay that the adversary causes. Since there exists some $v \in S$ that is at least $\ceil{\frac{n - t}{2}} \geq t + 1$ honest parties' input, we know that by the time $T + \Delta$ every honest party receives $t + 1$ echoes on $v$ and thus echoes $v$ even if $v$ is not its input. Then, by the time $T + 2\Delta$, every honest party receives $n - t \geq 2t + 1$ echoes on $v_1$. Finally, this is followed by one of the two cases below, in both of which the honest parties \mbox{all output by the time $T + 3\Delta$.} \begin{enumerate}
    \item It can be the case that no honest party proposes any $v' \neq v$ by the time $T + 2\Delta$. Then, by the time $T + 2\Delta$ every honest party proposes $v$ after receiving $2t + 1$ echoes on $v$. Consequently, by the time $T + 3\Delta$, every honest party receives $n - t$ proposals on $v$, which allows the party to output $v$ if it did not output earlier.
    \item It can be the case that some honest party $P_i$ proposes some $v' \neq v$ by the time $T + 2\Delta$. Then, at the time $T + 2\Delta$ the party $P_i$ must have received $2t + 1$ echoes on $v'$, with at least $t + 1$ of these echoes being from honest parties. Since both $v \in S$ and $v' \in S$ get echoed by at least $t + 1$ honest parties by the time $T + 2\Delta$, we conclude that by the time $T + 3\Delta$ every honest party $P_j$ gains the ability to output $V_j = \{v, v'\} = S$ after it receives $t + 1$ echoes on both $v$ and $v'$ and thus adds \mbox{both $v$ and $v'$ to $V_j$.} \qedhere
\end{enumerate}

\subparagraph{Complexity.} The round complexity is $3$, as proven above. The message complexity is $\BO(n^2)$ as each party multicasts at most $3$ messages (two echoes and one proposal), with each honest party echo/proposal carrying an honest party's input. Again, the communication is balanced.

\end{proof}

With $\mathsf{Prop}$, it is simple to grade-double $\GC_{2^{k-1}}$ into $\GC_{2^{k}}$. All we need to do is sequentially compose $\GC_{2^{k-1}}$ and $\mathsf{Prop}$, and interpret $\mathsf{Prop}$ outputs as $\GC_{2^k}$ outputs.

\begin{algobox}[\textbf{ with the input }$v_i$]{Protocol $\mathsf{GC_{2^k}}$ when $k \geq 1$}
    \State join a common instance of $\GC_{2^{k - 1}}$ with the input $v_i$
    \State join a common instance of $\mathsf{Prop}$
    \Upon{outputting $(y, g)$ from $\GC_{2^{k-1}}$}
        \State let your $\Prop$ input be $(y, g)$
    \EndUpon
    \Upon{outputting $Y$ from $\mathsf{Prop}$}
        \State if $Y = \{(y, j)\}$ for some $y$ and some $j \geq 0$, then output $(y, 2j)$
        \State if $Y = \{(y, j), (y', j + 1)\}$ for some $y$, $y'$ and some $j \geq 0$, then output $(y', 2j+1)$
    \EndUpon
\end{algobox}

\begin{theorem}
    For all $k \geq 0$, $\GC_{2^k}$ is a secure $2^k$-graded consensus protocol with liveness.
\end{theorem}
\begin{proof}
    Naturally, this theorem follows by induction on $k$. The base case $k = 0$ is the protocol $\GC_1$, which we have already proven secure. Below, we consider $k \geq 1$, where we obtain $\GC_{2^k}$ by sequentially composing $\GC_{2^{k-1}}$ and $\Prop$.

    The honest parties begin by running $\GC_{2^{k-1}}$ with their inputs. So, for some honest input $v^*$ and some grade $g \in \{0, \dots, 2^{k-1} - 1\}$ they each output either $(v',g)$ or $(v^*,g+1)$ from $\GC_{2^{k-1}}$, where $v' = \bot$ if $g = 0$ and $v' = v^*$ if $g \geq 1$. Afterwards, they run $\Prop$ with inputs in $\{(v',g), (v^*,g+1)\}$, and therefore either they all output $\{(v', g)\}$ or $\{(v', g), (v^*, g + 1)\}$ from $\Prop$ and thus all output $(v', 2g)$ or $(v^*, 2g+1)$ from $\GC_{2^k}$, or they all output $\{(v', g), (v^*, g + 1)\}$ or $\{(v^*, g + 1)\}$ from $\Prop$ and thus all output $(v^*, 2g+1)$ or $(v^*, 2g+2)$ from $\GC_{2^k}$. Either way, $\GC_{2^k}$ achieves liveness, intrusion tolerance and agreement.

    If the honest parties run $\GC_{2^k}$ with a common input $v^*$, then they all run $\GC_{2^{k-1}}$ with the input $v^*$, output $(v^*, 2^{k-1})$ from it, run $\Prop$ with the input $(v^*, 2^{k-1})$, output $\{(v^*, 2^{k-1})\}$ from it, and finally output $(v^*, 2^k)$ from $\GC_{2^k}$. Hence, $\GC_{2^k}$ achieves validity as well. \qedhere
\end{proof}

\subparagraph*{Complexity.} As $\GC_{2^k}$ consists of one 3-round $\GC_1$ instance followed by $k$ sequential 3-round $\Prop$ iterations, $\GC_{2^k}$'s round complexity is $3k + 3$ and its message complexity is $\BO(kn^2)$. When the parties run $\GC_{2^k}$ with $\ell$-bit inputs, they run $\GC_1$ with inputs in $\{0, 1\}^\ell$, and so $\GC_1$ costs $\BO(\ell n^2)$ bits of communication. Then, for each $j \in \{1, \dots, k\}$ they run the $j^{\text{th}}$ $\Prop$ iteration with inputs that are $2^{j-1}$-graded consensus outputs in $\{0, 1\}^\ell \times \{1, \dots, 2^{j-1}\} \cup \{(\bot, 0)\}$, of size $\BO(\ell + \log j) = \BO(\ell + \log k)$, and so the $j^{\text{th}}$ $\Prop$ iteration costs $\BO(n^2(\ell + \log k))$ bits of communication. Overall, we see that $\GC_{2^k}$'s communication complexity is $\BO(kn^2(\ell + \log k))$ bits, with the $\log k$ term here also covering the iteration IDs for the $\Prop$ iterations. Note that the communication \mbox{in $\GC_{2^k}$ is balanced, as both $\GC_1$ and $\Prop$ have balanced communication.}

\section{\texorpdfstring{The Termination Protocol $\Term$}{The Termination Protocol Term}} \label{termAppendix}

\begin{algobox}{Protocol $\Term$}
    \State $y_i \gets \bot$ \Comment{Here, $\bot$ is a placeholder value than cannot be a valid $\Term$ input.}
    \For{any $v$, if you acquire the input $v$ or receive $\msgpair{\ECHO}{v}$ from $t + 1$ parties}
        \State if $y_i = \bot$, then $y_i \gets v$
        \State multicast $\msgpair{\ECHO}{v}$
    \EndFor

    \When{you have received $\READY$ from $t + 1$ parties or there exists some $v$ such that you have received $\msgpair{\ECHO}{v}$ from $2t + 1$ parties}
        \State multicast $\READY$
    \EndWhen

    \When{you have received $\READY$ from $2t + 1$ parties, you have multicast $\READY$ and you have set $y_i \neq \bot$}
        \State output $y_i$ and terminate
    \EndWhen
\end{algobox}

\termSec

\begin{proof}
    Suppose an honest party $P_i$ outputs some $v$ from $\Term$. Then, $P_i$ has received $t + 1$ echoes on $v$, at least one of which is honest. An honest party only echoes a value if the value is its input or it receives $t + 1$ echoes on the value (at least one of which is honest), which means that the first honest party who echoed $v$ did so because its input is $v$. So, we see that the honest parties can only output honest party inputs from $\Term$.

    As for termination in $3$ rounds, let $\Delta$ be the maximum honest message delay the adversary causes. We first prove a fact we will use later: If every honest party multicasts $\READY$ by some time $T$, then every honest party terminates by the time $T + \Delta$. To see why, suppose every honest party multicasts $\READY$ by some time $T$. At the time $T$, the first honest party who has multicast $\READY$ has done so because it has received $2t + 1$ echoes on some $v$, at least $t + 1$ of which are honest. So, by the time $T + \Delta$ every honest party $P_i$ receives $t + 1$ honest echoes on $v$, which allows $P_i$ to set $y_i \gets v$ if $y_i = \bot$. Moreover, since every honest party multicasts $\READY$ by the time $T$, every honest party $P_i$ receives $n - t \geq 2t + 1$ honest $\READY$ messages by the time $T + \Delta$, \mbox{which allows $P_i$ to terminate $\Term$ by the time $T + \Delta$ with an output $y_i \neq \bot$.}

    Now, let us show that if an honest party $P_i$ terminates by some time $T$, then every honest party terminates by the the time $T + 2\Delta$. This is because if $P_i$ terminates by the time $T$, it does so after receiving $\READY$ from $2t + 1$ parties, at least $t + 1$ of which are honest. Since at least $t + 1$ honest parties multicast $\READY$ by the time $T$, every honest party receives $t + 1$ honest $\READY$ messages and therefore multicasts $\READY$ by the time $T + \Delta$, and consequently every honest party terminates by the time $T + 2\Delta$.

    Finally, suppose that by some time $T$ every honest party acquires an input. We want to show that every honest party terminates by the time $T + 3\Delta$. If some honest party terminates by the time $T + \Delta$, then this follows from what we have proven above; so, suppose no honest party terminates by the time $T + \Delta$. By the time $T$, every honest party echoes its input in $\{y, y'\}$, and so there exists some $y^* \in \{y, y'\}$ that at least $\ceil{\frac{n - t}{2}} \geq t + 1$ honest parties echo. By the time $T + \Delta$, every honest party receives $t + 1$ echoes on $y^*$ before terminating, and thus echoes $y^*$ even if it has not acquired the input $y^*$. By the time $T + 2\Delta$, every honest party receives $n - t \geq 2t + 1$ echoes on $y^*$ and thus multicasts $\READY$. Finally by the time $T + 3\Delta$, every honest party terminates $\Term$. \qedhere
\end{proof}

\section{Skipped Proofs}

\expsecurity

\begin{proof}
    To prove the theorem, we pick any arbitrary $q \geq 1$, and by induction prove for each $j \in \{q - 1, q - 2, \dots, 0\}$ (in this order) that $\Exp_j$ is a protocol for edge agreement in the path $(2^j - 1, \dots, 2^q - 1)$ that takes at most $6(j + 1)$ rounds if run with the common input $2^j - 1$ and at most $7 + 12q - 6j$ rounds otherwise. 
    
    Note that for any $j \geq 0$, if no honest party obtains the grade $0$ from $\GC_2$ in $\Exp_j$, then no honest party multicasts $\CENTER$, and thus no honest party outputs $2^{j +1} - 1$ on Line \ref{line:exp-center-out} upon receiving $\CENTER$ from $t + 1$ parties. We use this fact in the rest of the proof below.

    The induction's base case is $q = j + 1 \geq 1$, where each honest party $P_i$ runs $\Exp_{j}$ with an input $v_i \in \{2^j - 1, \dots, 2^{j + 1} - 1\}$. In this base case, the honest parties all input $\LEFT$ to $\GC_2$, and thus all output $(\LEFT, 2)$ from it. Then, each honest party $P_i$ sets $\mathsf{side} = \LEFT$, sets $v_i^{\sfnext} = \min(v_i, 2^{j + 1} - 1) = v_i$, runs $\TC((2^j - 1, \dots, 2^{j + 1} - 1))$ with the other parties with the input $v_i^{\sfnext} = v_i$, and finally lets its $\TC$ output be its final output. So, $\Exp_j$'s security follows from $\TC$'s security, and $\Exp_j$ takes at most $6 + 6j + 1 = 1 + 12q - 6j < 7 + 12q - 6j$ rounds: $6$ rounds due to $\GC_2$ and $6j + 1$ rounds due to $\TC$ (since $\TC$ is run on a path of length $2^j$). However, if the honest parties have a common input, then $\TC$ only takes $6j$ rounds, and thus $\Exp_j$ takes $6(j + 1)$ rounds instead of $6(j + 1) + 1$.

    Now, let us consider the cases that might arise when the honest parties run $\Exp_j$ with inputs in $\{2^j - 1, \dots, 2^{q} - 1\}$ for any $j \in \{q-2, q-3, \dots, 0\}$, with the inductive assumption that $\Exp_{j + 1}$ provides edge agreement in $(2^{j + 1} - 1, 2^q - 1)$ in at most $6(j + 2)$ rounds when run with the common input $2^{j + 1} - 1$ and at most $7 + 12q - 6(j + 1)$ rounds otherwise.\begin{enumerate}
        \item It could be the case that every honest input is less than $2^{j + 1}$. This case is identical to the base case above; the honest parties reach edge agreement in at most $6(j+1)$ rounds if they have a common input and in at most $6(j+1) + 1 < 7 + 12q - 6j$ rounds otherwise.
        \item It could be the case that every honest input is at least $2^{j + 1}$. Then, the honest parties all run $\GC_2$ with the input $\RIGHT$, and thus output $(\RIGHT, 2)$ from it. Then, each honest party $P_i$ sets $\mathsf{side} = \RIGHT$, sets $v_i^{\sfnext} = \max(v_i, 2^{j+1} - 1) = v_i$, runs $\Exp_{j + 1}$ with the other parties for edge agreement in the path $(2^{j + 1} - 1, \dots, 2^q - 1)$ with the $\Exp_{j + 1}$ input $v_i^{\sfnext} = v_i$, and finally lets its $\Exp_{j + 1}$ output be its final output. Hence, $\Exp_j$'s security follows from $\Exp_{j+1}$'s security, and $\Exp_j$ takes at most $7 + 12q - 6j$ rounds: $6$ rounds due to the initial $\GC_2$, and \mbox{$\max(6(j+2), 7 + 12q - 6(j + 1)) = 7 + 12q - 6(j + 1)$ rounds due to $\Exp_{j+1}$.} The last inequality here follows from the fact that $q \geq j + 2$.
        \item The most challenging case is when some honest inputs are below $2^{j + 1}$, while others or not. In this case, the parties run $\GC_2$ with different inputs, and so for some $k \in \{\LEFT, \RIGHT\}$ they either all output in $(k, 2)$ or $(k, 1)$ from $\GC_2$, or all output $(k, 1)$ or $(\bot, 0)$ from $\GC_2$. Below, we say that the corresponding protocol for the graded consensus value $k = \LEFT$ is $\TC((2^j - 1, \dots, 2^{j + 1} - 1))$, while the corresponding protocol for $k = \RIGHT$ is $\Exp_{j + 1}$. \begin{alphaenumerate}
            \item The case where the parties output $(k, 2)$ or $(k, 1)$ from $\GC_2$ for some $k \in \{\LEFT, \RIGHT\}$ is like the two cases we considered above where they all output $(k, 2)$ from $\GC_2$, except for that some honest parties run the corresponding protocol with the input $v_i^{\sfnext} = v_i$ while some others run it with the input $v_i^{\sfnext} = 2^{j+1} - 1$. As some honest inputs are below $2^{j+1} - 1$ while others or not, $2^{j+1} - 1$ is in the convex hull of the honest inputs, which means that some honest parties switching their inputs to it does not impact convex validity. So, the honest parties reach edge agreement in at most $7 + 12q - 6j$ rounds by running $\GC_2$ in $6$ rounds and then running the corresponding protocol in at most \mbox{$\max(6j + 1, \max(6(j+2), 7 + 12q - 6(j + 1))) = 7 + 12q - 6(j + 1)$ rounds.}
            \item If the honest parties all output $(k, 1)$ or $(k, \bot)$ from $\GC_2$ for some $k \in \{\LEFT, \RIGHT\}$, then every honest party $P_i$ sets $v_i^{\sfnext} = 2^{j + 1} - 1$, which ensures that the honest parties can only run $\TC$ or $\Exp_{j+1}$ with the input $2^{j + 1} - 1$ and thus only output $2^{j + 1} - 1$ from them. Hence, an honest party that outputs from $\Exp_j$ can only output the valid value $2^{j + 1} - 1$, no matter if it outputs on Line \ref{line:exp-out-bot}, \ref{line:exp-out-le}, \ref{line:exp-out-ri} or \ref{line:exp-center-out}. What remains to show is liveness. Observe that since $n - t \geq 2t + 1$, there are either $t + 1$ honest parties that output $(\bot, 0)$ from $\GC_2$, or $t + 1$ honest parties that output \nolinebreak $(k, 1)$ \nolinebreak from \nolinebreak $\GC_2$. \begin{romanenumerate}
                \item In the former case, the $t + 1$ honest parties that output $(\bot, 0)$ from $\GC_2$ multicast $\CENTER$, and so, after one round (i.e.\ $7 \leq 7 + 12q - 6j$ rounds after $\Exp_j$ begins) every honest party becomes able to output $2^{j + 1} - 1$ on Line \ref{line:exp-center-out}.
                \item In the latter case, the $t + 1$ honest parties that output $(k, 1)$ from $\GC_2$ multicast $k \in \{\LEFT, \RIGHT\}$ before running the corresponding protocol, and this allows the honest parties that output $(\bot, 0)$ from $\GC_2$ to also set $\mathsf{side} = k$ and run the corresponding protocol, albeit with one extra round of delay compared to the parties that output $(k, 1)$ from $\GC_2$. Therefore, every honest party that does not output $2^{j + 1} - 1$ on Line \ref{line:exp-center-out} eventually outputs $2^{j + 1} - 1$ by outputting this from the corresponding protocol. The corresponding protocol (which the honest parties run with the common input $2^{j+1} - 1$) is either $\TC$, which takes at most $6j$ rounds, or $\Exp_{j + 1}$, which takes at most $6(j+2)$ rounds. So, in total, every honest party outputs in at most $6 + 1 + 6(j + 2) \leq 7 + 12q - 6j$ rounds. Here, the first $6$ rounds are due to $\GC_2$, the following $7^{\text{th}}$ round is due to the honest parties that output $(\bot, 0)$ from $\GC_2$ needing one extra round to learn $k$, and the final $6(j + 2) \leq 12q - 6j$ rounds are due to the corresponding protocol. \qedhere
            \end{romanenumerate}
        \end{alphaenumerate}
    \end{enumerate}
\end{proof}

\dexpsecurity

\begin{proof}
    Let $V = \{\floor{\log_2(v_i + 1)} : P_i \text{ is honest}\}$, let $q_1 = \min V$, and let $q_2 = \max V$. Observe that $v_i \in \{2^{q_1} - 1, \dots, 2^{q_2 + 1} - 1\}$ for every honest party $P_i$ and that $q_2 = \floor{\log _ 2 (M + 1)}$ if the maximum honest input is $M$. The honest parties run $\Agn$ with inputs in $\{5q_1, \dots, 5q_2\}$, and so they output adjacent integers in $\{5q_1, \dots, 5q_2\}$ from it, in at most $f(5q_2)$ rounds.

    If $q_1 = q_2$, then the honest parties all run $\Agn$ with the common input $5q_2$, and thus they all output $5q_2$ from it. Consequently, each honest party $P_i$ sets $(k_i, r_i) = (q_2, 0)$, and obtains its final output from an instance of $\TC((2^{q_2} - 1, \dots, 2^{q_2 + 1} - 1))$ which it runs with the input $v_i^{\sfnext} = \min(\max(v_i, 2^{q_2} - 1), 2^{q_2 + 1} - 1) = v_i \in \{2^{q_2} - 1, \dots, 2^{q_2 + 1} - 1\}$. Edge agreement thus follows from $\TC$, in $6q_2 + 1$ more rounds.

    For the rest of the proof, we assume $q_1 < q_2$. Observe that the boundary values $q_1$ and $q_2$ arise from some honest inputs $v_i, v_j$ being such that $\floor{\log_2(v_i + 1)} = q_1$ and $\floor{\log_2(v_j + 1)} = q_2$. These equalities respectively imply $v_i < 2^{q_1 + 1} - 1$ and $v_j \geq 2^{q_2} - 1$, which means that every integer in $\{2^{q_1 + 1} - 1, \dots, 2^{q_2} - 1\}$ is in the convex hull of the honest parties' inputs.

    When an honest party $P_i$ outputs $5k_i + r_i$ from $\Agn$, it lets $v_i^{\sfnext} = \min(\max(v_i, 2^{k_i} - 1), 2^{k_i + 1} - 1)$ if $r_i = 0$, and lets $v_i^{\sfnext} = 2^{k_i + 1} - 1$ otherwise. Let us show that in every possible case $P_i$ chooses a valid $v_i^{\sfnext}$; that is, $P_i$ sets $v_i^{\sfnext}$ to a value that is in the convex hull of the honest inputs. If $r_i = 0$, then there are three cases: If $v_i < 2^{k_i} - 1$ then $v_i^{\sfnext} = 2^{k_i} - 1$ is valid as $2^{k_i} - 1$ is between the valid values $v_i$ and $2^{q_2} - 1$, if $v_i > 2^{k_i + 1} - 1$ then $v_i^{\sfnext} = 2^{k_i + 1} - 1$ is valid as $2^{k_i + 1} - 1$ is between the valid values $2^{q_1 + 1} - 1$ and $v_i$, and finally if $2^{k_i} - 1 \leq v_i \leq 2^{k_i + 1} - 1$ then $v_i^{\sfnext} = v_i$ is valid as well. Meanwhile, if $r_i > 0$, then by $5q_1 \leq 5k_i + r_i \leq 5q_2$ we have $q_1 \leq k_i \leq q_2 - 1$, which makes $v_i^{\sfnext} = 2^{k_i + 1} - 1 \in \{2^{q_1 + 1} - 1, \dots, 2^{q_2} - 1\}$ valid.
    
    Convex validity follows from the fact that every honest party $P_i$ either directly outputs the valid value $2^{k_i + 1} - 1$ or obtains its output from a $\TC$ instance for which the honest parties can only acquire valid inputs. Note that the honest parties do not run $\TC$ with out-of-range inputs: The value $2^{k_i + 1} - 1$ which a party $P_i$ inputs to $\TC((2^{k_i} - 1, \dots, 2^{k_i + 1} - 1))$ if $r_i \in \{1, 2\}$ or $\TC((2^{k_i + 1} - 1, \dots, 2^{k_i + 2} - 1))$ if $r_i \in \{3, 4\}$ is in range for both $\TC((2^{k_i} - 1, \dots, 2^{k_i + 1} - 1))$ and $\TC((2^{k_i + 1} - 1, \dots, 2^{k_i + 1} - 2))$, and the value $\min(\max(v_i, 2^{k_i} - 1), 2^{k_i + 1} - 1)$ which a party $P_i$ inputs to $\TC((2^{k_i} - 1, \dots, 2^{k_i + 1} - 1))$ if $r_i = 0$ is in-range as well.

    The initial $\Agn$ guarantees that there exists some $v = 5k + r$ where $q_1 \leq k \leq q_2 - 1$ and $0 \leq r \leq 4$ such that every honest party outputs either $v$ or $v + 1$ from it. Based on this, let us prove the edge agreement and liveness properties of $\TStep(\Agn)$ via case analysis on $r$. \begin{enumerate}
        \item If $r = 0$, then the honest parties all run $\TC((2^k - 1, \dots, 2^{k+1} - 1))$ with in-range valid inputs, and obtain their final outputs from it. Edge agreement in $\mathbb{N}$ thus follows from $\TC((2^k - 1, \dots, 2^{k+1} - 1))$, in $6k + 1$ more rounds.
        \item If $r = 1$, then the honest parties run $\TC((2^k - 1, \dots, 2^{k+1} - 1))$ with the common input $2^{k+1} - 1$, and thus all output $2^{k+1} - 1$ from it in $6k$ rounds. The honest parties $P_i$ that have $r_i = 1$ output $2^{k + 1} - 1$ from $\TStep(\Agn)$ after they output $2^{k + 1} - 1$ from $\TC$, while those that have $r_i = 2$ output $2^{k + 1} - 1$ directly after they output from \nolinebreak $\Agn$.
        \item If $r = 2$, then every honest party $P_i$ has $r_i \in \{2, 3\}$, which leads to it directly outputting $2^{k + 1} - 1$ after it outputs from $\Agn$. The fact that the honest parties $P_i$ with $r_i = 2$ run $\TC((2^k - 1, \dots, 2^{k+1} - 1))$ while those with $r_i = 3$ run $\TC((2^{k+1} - 1, \dots, 2^{k+2} - 1))$ is not an issue as none of the honest parties care about their $\TC$ outputs.
        \item If $r = 3$, then the honest parties run $\TC((2^{k+1} - 1, \dots, 2^{k+2} - 1))$ with the common input $2^{k+1} - 1$, and thus all output $2^{k+1} - 1$ from it in $6(k+1)$ rounds. The honest parties $P_i$ that have $r_i = 4$ output $2^{k + 1} - 1$ from $\TStep(\Agn)$ after they output $2^{k + 1} - 1$ from $\TC$, while those that have $r_i = 3$ output $2^{k + 1} - 1$ directly after they output from $\Agn$.
        \item If $r = 4$, then every honest party $P_i$ (both if $(k_i, r_i) = (k, r)$ and if $(k_i, r_i) = (k + 1, 0)$) runs $\TC((2^{k+1} - 1, \dots, 2^{k+2} - 1))$ with an in-range valid input, and obtains its final output from it. Edge agreement thus follows from $\TC((2^{k+1} - 1, \dots, 2^{k+2} - 1))$, in $6(k+1)+1$ more rounds.
    \end{enumerate} 
Note that in all the cases above, $\TStep(\Agn)$ takes at most $6q_2+1$ rounds after $\Agn$. Hence, $\TStep(\Agn)$ takes at most $f(5q_2) + 6q_2+1$ rounds. This is the round complexity we wanted to prove since $q_2 = \ceil{\log_2(M + 1)}$.
\end{proof}

\agntoagz

\begin{proof}
    If the honest parties all have inputs in $\mathbb{N}$, then every honest party $P_i$ runs $\GC_2$ with the input $1$, outputs $(1, 2)$ from $\GC_2$, lets $v_i^{\sfnext} = \max(0, (2 - 1)(1)v_i) = v_i$, runs $\Agn$ with the input $v_i^{\sfnext} = v_i$, and finally outputs $1y = y$ from $\Agz(\Agn)$ when it outputs $y$ from $\Agn$. Thus, $\Agz(\Agn)$' security follows from $\Agn$'s security.

    Similarly, if the honest parties all have negative inputs, then every honest party $P_i$ runs $\GC_2$ with the input $-1$, outputs $(-1, 2)$ from $\GC_2$, lets $v_i^{\sfnext} = \max(0, (2 - 1)(-1)v_i) = -v_i$, runs $\Agn$ with the input $v_i^{\sfnext} = -v_i$, and finally outputs $-1y = -y$ from $\Agz(\Agn)$ when it outputs $y$ from $\Agn$. Again, $\Agz(\Agn)$'s security follows from $\Agn$'s security. This is because $\Agn$ becomes an edge agreement protocol for $(\dots, -1, 0)$ when the parties mirror (sign-flip) their inputs and outputs, by symmetry.

    Finally, there is the case where some honest parties have negative inputs while others do not. In this case, the honest parties run $\GC_2$ with different inputs, and so for some $k \in \{1, -1\}$ they either all output $(k, 2)$ or $(k, 1)$ from $\GC_2$ or all output $(k, 1)$ or $(k, \bot)$ from $\GC_2$. The case where they all output $(k, 2)$ or $(k, 1)$ from $\GC_2$ is similar to the cases above, though with some honest parties $P_i$ running $\Agn$ with the input $v_i^{\sfnext} = 0kv_i = 0$ rather than $v_i^{\sfnext} = 1kv_i = kv_i$. As $0$ is in the convex hull of the honest inputs, this does not impact validity, and the honest parties safely reach edge agreement via $\Agn$, with input/output mirroring if $k = -1$ and without mirroring if $k = 1$. Meanwhile, if the honest parties all output $(k, 1)$ or $(\bot, 0)$ from $\GC_2$, then they all run $\Agn$ with the input $0$. So, the honest parties with the $\GC_2$ output $(k, 1)$ output $k0 = 0$ from $\Agz(\Agn)$ after outputting $0$ from $\Agn$, while those with the $\GC_2$ output $(\bot, 0)$ output $0$ \mbox{from $\Agz(\Agn)$ directly without caring about their $\Agn$ outputs.} \qedhere   
\end{proof}

\end{document}